\documentclass[a4paper,UKenglish,cleveref, autoref, thm-restate]{lipics-v2021}\nolinenumbers
\usepackage[T1]{fontenc}
\usepackage{graphicx}
\usepackage{booktabs}   \usepackage{subcaption} \usepackage{bussproofs}
\usepackage[cal=boondoxo]{mathalfa}
\DeclareMathAlphabet{\mathpzc}{OT1}{pzc}{m}{it}
\usepackage{stmaryrd}
\usepackage{amsmath}
\usepackage{color}
\usepackage{xspace}

\newcommand{\ignore}[1]{{}}

\newcommand{\syntaxDef}[3]{\rulebox{\syntaxKeyword$#1\mathrel{::=}{#2}$ \ifthenelse{\equal{#3}{}}{}{[#3]}}}

\makeatletter
\newcommand{\shorteq}{\settowidth{\@tempdima}{-}\resizebox{\@tempdima}{\height}{=}}
\makeatother

 \usepackage{bm}
\usepackage{amssymb}
\usepackage{xcolor}
\usepackage{colortbl}

\usepackage[utf8]{inputenc}

\usepackage{amsmath}
\usepackage{hyperref}

\usepackage{subcaption}
\usepackage{wrapfig}

\usepackage{xspace}
\usepackage{float}
\usepackage{balance}

\usepackage{textcomp} 

\usepackage{changepage}
\usepackage{array,etoolbox}

\preto\tabular{\setcounter{magicrownumbers}{0}}
\newcounter{magicrownumbers}

\usepackage{cleveref} 

\usepackage{wrapfig}
\usepackage{caption}

\usepackage{longtable}
\usepackage{tabu}

\usepackage{mathpartir}

\usepackage{mathtools} \usepackage{xfrac}

\usepackage[qm,braket]{qcircuit}

\newcommand{\sqir}{\textsc{sqir}\xspace}

\newcommand{\myparagraph}[1]{\paragraph{\textbf{#1.}}}

\newcommand{\CC}{\ensuremath{\mathbb{C}}\xspace}

\usepackage{pgf}

\usepackage{tikz} \usetikzlibrary{arrows,shapes.misc,shapes.geometric, shapes.arrows,shapes.callouts,
  shapes.gates.logic.US,
  chains,matrix,positioning,scopes,decorations.pathmorphing,decorations.text,
  decorations.pathreplacing, shadows,automata,
  fit, calc, arrows.meta,
  trees
}

\tikzset{ machine/.style={
rectangle,
minimum width=25mm,
    minimum height=18mm,
    text width=24mm,
align=center,
very thick,
    draw=black,
color=black,
    fill=white,
}
}

\DeclarePairedDelimiter\abs{\lvert}{\rvert}
\DeclarePairedDelimiter\norm{\lVert}{\rVert}

\makeatletter
\let\oldabs\abs
\def\abs{\@ifstar{\oldabs}{\oldabs*}}
\let\oldnorm\norm
\def\norm{\@ifstar{\oldnorm}{\oldnorm*}}
\makeatother

\makeatletter
\DeclareRobustCommand{\vardivision}{\mathbin{\mathpalette\@vardivision\relax}}
\newcommand{\@vardivision}[2]{\reflectbox{$\m@th\smallsetminus$}}
\makeatother

\usepackage{booktabs}

\usepackage{alltt}
\usepackage{listings,lstcoq}
\definecolor{ltblue}{rgb}{0,0.4,0.4}
\definecolor{dkblue}{rgb}{0,0.1,0.6}
\definecolor{dkgreen}{rgb}{0,0.35,0}
\definecolor{dkviolet}{rgb}{0.3,0,0.5}
\definecolor{dkred}{rgb}{0.5,0,0}
\usepackage[export]{adjustbox}

\newcommand{\code}[1]{{\small\texttt{#1}}}

\newcommand{\rulelab}[1]{{\small \textsc{#1}}}

\usepackage{newunicodechar}
\let\Alpha=A
\let\Beta=B
\let\Epsilon=E
\let\Zeta=Z
\let\Eta=H
\let\Iota=I
\let\Kappa=K
\let\Mu=M
\let\Nu=N
\let\Omicron=O
\let\omicron=o
\let\Rho=P
\let\Tau=T
\let\Chi=X

\newunicodechar{Α}{\ensuremath{\Alpha}}
\newunicodechar{α}{\ensuremath{\alpha}}
\newunicodechar{Β}{\ensuremath{\Beta}}
\newunicodechar{β}{\ensuremath{\beta}}
\newunicodechar{Γ}{\ensuremath{\Gamma}}
\newunicodechar{γ}{\ensuremath{\gamma}}
\newunicodechar{Δ}{\ensuremath{\Delta}}
\newunicodechar{δ}{\ensuremath{\delta}}
\newunicodechar{Ε}{\ensuremath{\Epsilon}}
\newunicodechar{ε}{\ensuremath{\epsilon}}
\newunicodechar{ϵ}{\ensuremath{\varepsilon}}
\newunicodechar{Ζ}{\ensuremath{\Zeta}}
\newunicodechar{ζ}{\ensuremath{\zeta}}
\newunicodechar{Η}{\ensuremath{\Eta}}
\newunicodechar{η}{\ensuremath{\eta}}
\newunicodechar{Θ}{\ensuremath{\Theta}}
\newunicodechar{θ}{\ensuremath{\theta}}
\newunicodechar{ϑ}{\ensuremath{\vartheta}}
\newunicodechar{Ι}{\ensuremath{\Iota}}
\newunicodechar{ι}{\ensuremath{\iota}}
\newunicodechar{Κ}{\ensuremath{\Kappa}}
\newunicodechar{κ}{\ensuremath{\kappa}}
\newunicodechar{Λ}{\ensuremath{\Lambda}}
\newunicodechar{λ}{\ensuremath{\lambda}}
\newunicodechar{Μ}{\ensuremath{\Mu}}
\newunicodechar{μ}{\ensuremath{\mu}}
\newunicodechar{Ν}{\ensuremath{\Nu}}
\newunicodechar{ν}{\ensuremath{\nu}}
\newunicodechar{Ξ}{\ensuremath{\Xi}}
\newunicodechar{ξ}{\ensuremath{\xi}}
\newunicodechar{Ο}{\ensuremath{\Omicron}}
\newunicodechar{ο}{\ensuremath{\omicron}}
\newunicodechar{Π}{\ensuremath{\Pi}}
\newunicodechar{π}{\ensuremath{\pi}}
\newunicodechar{ϖ}{\ensuremath{\varpi}}
\newunicodechar{Ρ}{\ensuremath{\Rho}}
\newunicodechar{ρ}{\ensuremath{\rho}}
\newunicodechar{ϱ}{\ensuremath{\varrho}}
\newunicodechar{Σ}{\ensuremath{\Sigma}}
\newunicodechar{σ}{\ensuremath{\sigma}}
\newunicodechar{ς}{\ensuremath{\varsigma}}
\newunicodechar{Τ}{\ensuremath{\Tau}}
\newunicodechar{τ}{\ensuremath{\tau}}
\newunicodechar{Υ}{\ensuremath{\Upsilon}}
\newunicodechar{υ}{\ensuremath{\upsilon}}
\newunicodechar{Φ}{\ensuremath{\Phi}}
\newunicodechar{φ}{\ensuremath{\phi}}
\newunicodechar{ϕ}{\ensuremath{\varphi}}
\newunicodechar{Χ}{\ensuremath{\Chi}}
\newunicodechar{χ}{\ensuremath{\chi}}
\newunicodechar{Ψ}{\ensuremath{\Psi}}
\newunicodechar{ψ}{\ensuremath{\psi}}
\newunicodechar{Ω}{\ensuremath{\Omega}}
\newunicodechar{ω}{\ensuremath{\omega}}

\newunicodechar{ℕ}{\ensuremath{\mathbb{N}}}
\newunicodechar{∅}{\ensuremath{\emptyset}}

\newunicodechar{∙}{\ensuremath{\bullet}}
\newunicodechar{≈}{\ensuremath{\approx}}
\newunicodechar{≅}{\ensuremath{\cong}}
\newunicodechar{≡}{\ensuremath{\equiv}}
\newunicodechar{≤}{\ensuremath{\le}}
\newunicodechar{≥}{\ensuremath{\ge}}
\newunicodechar{≠}{\ensuremath{\neq}}
\newunicodechar{∀}{\ensuremath{\forall}}
\newunicodechar{∃}{\ensuremath{\exists}}
\newunicodechar{±}{\ensuremath{\pm}}
\newunicodechar{∓}{\ensuremath{\pm}}
\newunicodechar{·}{\ensuremath{\cdot}}
\newunicodechar{⋯}{\ensuremath{\cdots}}
\newunicodechar{…}{\ensuremath{\ldots}}
\newunicodechar{∷}{~\mathrel{:\!\!\!:}~}
\newunicodechar{×}{\ensuremath{\times}}
\newunicodechar{∞}{\ensuremath{\infty}}
\newunicodechar{→}{\ensuremath{\to}}
\newunicodechar{←}{\ensuremath{\leftarrow}}
\newunicodechar{⇒}{\ensuremath{\Rightarrow}}
\newunicodechar{↦}{\ensuremath{\mapsto}}
\newunicodechar{↝}{\ensuremath{\leadsto}}
\newunicodechar{∨}{\ensuremath{\vee}}
\newunicodechar{∧}{\ensuremath{\wedge}}
\newunicodechar{⊢}{\ensuremath{\vdash}}
\newunicodechar{⊣}{\ensuremath{\dashv}}
\newunicodechar{∣}{\ensuremath{\mid}}
\newunicodechar{∈}{\ensuremath{\in}}
\newunicodechar{⊆}{\ensuremath{\subseteq}}
\newunicodechar{⊂}{\ensuremath{\subset}}
\newunicodechar{∪}{\ensuremath{\cup}}
\newunicodechar{⋓}{\ensuremath{\Cup}}
\newunicodechar{∉}{\ensuremath{\not\in}}
\newunicodechar{√}{\ensuremath{\sqrt}}

\newunicodechar{⊸}{\ensuremath{\multimap}}
\newunicodechar{⊗}{\ensuremath{\otimes}}
\newunicodechar{⨂}{\ensuremath{\bigotimes}}
\newunicodechar{⊕}{\ensuremath{\oplus}}
\newunicodechar{〈}{\ensuremath{\langle}}
\newunicodechar{⟨}{\ensuremath{\langle}}
\newunicodechar{⟩}{\ensuremath{\rangle}}
\newunicodechar{〉}{\ensuremath{\rangle}}
\newunicodechar{¡}{\ensuremath{\upsidedownbang}}
\newunicodechar{∘}{\ensuremath{\circ}}
\newunicodechar{†}{\ensuremath{\dagger}}
\newunicodechar{⊤}{\ensuremath{\top}}
\newunicodechar{⊥}{\ensuremath{\bot}}

\newunicodechar{〚}{\ensuremath{\llbracket}}
\newunicodechar{〛}{\ensuremath{\rrbracket}}

\usepackage{etoolbox} \newtoggle{comments}
\toggletrue{comments}

\iftoggle{comments}{
  \newcommand{\fixme}[1]{\textbf{\textcolor{red}{[ Fixme: #1]}}}
  \newcommand{\todo}[1]{\textbf{\textcolor{green}{[ TODO: #1 ]}}}
  \newcommand{\mwh}[1]{\textbf{\textcolor{red}{[ Mike: #1 ]}}}
  
\newcommand{\shh}[1]{\textbf{\textcolor{purple}{[ Shih-Han: #1 ]}}}
  \newcommand{\liyi}[1]{\textbf{\textcolor{blue}{[ Liyi: #1 ]}}}
  \newcommand{\oth}[2]{\textbf{\textcolor{red}{[ #1: #2 ]}}}
  \newcommand{\xwu}[1]{\textbf{\textcolor{purple}{[ Xiaodi: #1 ]}}}
  
}{
  \newcommand{\fixme}[1]{}
  \newcommand{\todo}[1]{}
  \newcommand{\rnr}[1]{}
  \newcommand{\mwh}[1]{}  
\newcommand{\liyi}[1]{}
  \newcommand{\shh}[1]{}
  \newcommand{\xwu}[1]{}
  \newcommand{\oth}[2]{}
}

\newtoggle{submission}
\toggletrue{submission}

\iftoggle{submission}{
  
}{
  
}

 \usepackage{pifont}

\colorlet{MZ}{violet!60!pink}
\usepackage[normalem]{ulem}
\newcommand{\mz}[1]{{\textcolor{MZ}{\textbf{[[}(\Walley[1.6][white!50!orange]) {\small{#1}}\textbf{]]}}}}

\newcommand{\was}[1]{}

\NewCommandCopy{\Creff}{\Cref}
\renewcommand{\Cref}[1]{\mbox{\Creff{#1}}}

\usepackage[inline]{enumitem}

\colorlet{lc}{teal}

\usepackage{tikzsymbols}
\usepackage[frozencache,cachedir=minted-cache]{minted}
\usetikzlibrary{calc,topaths}

\newcommand\wideparen[1]{\tikz[baseline=(wideArcAnchor.base)]{
    \node[inner sep=0] (wideArcAnchor) {$#1$}; 
    \coordinate (wideArcAnchorA) at ($0.9*(wideArcAnchor.north west) + 0.1*(wideArcAnchor.north east)+(0.0em,0.75ex)$);
    \coordinate (wideArcAnchorB) at ($0.1*(wideArcAnchor.north west) + 0.9*(wideArcAnchor.north east)+(0.0em,0.75ex)$);
\draw[line width=0.1ex,line cap=round] 
        ($(wideArcAnchor.north west)+(0.0em,0.1ex)$) 
            .. controls (wideArcAnchorA) and (wideArcAnchorB) ..
        ($(wideArcAnchor.north east)+(0.0em,0.1ex)$)        
    ;
}}

\newcommand{\pard}[2]{#1\texttt{,}\,#2}
\newcommand{\parp}[2]{#1\,+\,#2}
\makeatletter
\newcommand{\thickhline}{\noalign {\ifnum 0=`}\fi \hrule height 1pt
    \futurelet \reserved@a \@xhline
}
\newcolumntype{"}{@{\hskip\tabcolsep\vrule width 1pt\hskip\tabcolsep}}
\makeatother

\newcommand{\cmsg}[1]{\wideparen{#1}}
\newcommand{\comp}[1]{\texttt{!}#1}

\newcommand{\downa}[1]{\nu \,#1 \texttt{.}}
\newcommand{\upa}[1]{\uparrow #1 \texttt{.}}
\newcommand{\encode}[2]{#1\triangleleft#2\texttt{.}}

\newcommand{\encodes}[3]{#1\texttt{!}#2(#3)\blacktriangleright}
\newcommand{\decode}[2]{#1\triangleright\texttt{(}#2\texttt{)}\texttt{.}}
\newcommand{\decodes}[2]{#1\texttt{!}(#2)\blacktriangleleft}

\newcommand{\msga}[2]{#1\texttt{?}(#2){\tiny \succ}}
\newcommand{\up}[1]{\dagger#1}

\newcommand{\csenda}[2]{#1\texttt{!}#2\texttt{.}}

\newcommand{\creva}[2]{#1\texttt{?(}#2\texttt{)}\texttt{.}}

\newcommand{\scell}[1]{\{\!| #1 |\!\}}

\newcommand{\darrow}[1]{\longrightarrow_{\{#1\}}}

\newcommand{\bscell}[1]{\big{\{}\!\!\big{|} #1 \big{|}\!\!\big{\}}}

\newcommand{\seq}[2]{#1 #2}
\newcommand{\parll}[3]{#1\,{|\![}{#2}{]\!|}\,#3}

\newcommand{\Ls}{\mathbb{L}}
\newcommand{\Cs}{\mathcal{C}}
\newcommand{\As}{\mathcal{A}}

\newcommand{\Tts}{\mathcal{T}}

\newcommand{\Qs}{\mathcal{Q}}

\newcommand{\cn}[1]{\texttt{#1}}

\title{The Quantum Abstract Machine {\large (Extended Version)}}

\author{Liyi Li}{Iowa State University }{liyili2@iastate.edu}{https://orcid.org/0000-0001-8184-0244}{}

\author{Le Chang}{University of Maryland}{lchang21@umd.edu}{}{}

\author{Rance Cleaveland}{University of Maryland}{rance@cs.umd.edu}{}{}

\author{Mingwei Zhu}{University of Maryland}{mzhu1@umd.edu}{}{}

\author{Xiaodi Wu}{University of Maryland}{xwu@cs.umd.edu}{https://orcid.org/0000-0001-8877-9802}{}

\authorrunning{L. Li and L. Chang and R. Cleaveland and M. Zhu and X. Wu} 

\Copyright{Li et al.} 



\category{} 

\relatedversion{}

\EventEditors{}
\EventNoEds{2}
\EventLongTitle{}
\EventShortTitle{}
\EventAcronym{}
\EventYear{2024}
\EventDate{}
\EventLocation{}
\EventLogo{}
\SeriesVolume{}
\ArticleNo{}
\begin{document}

\title{The Quantum Abstract Machine}

\maketitle

\begin{abstract}
This paper develops a model of quantum behavior that is intended to support the abstract yet accurate design and functional verification of quantum communication protocols. The work is motivated by the need for conceptual tools for the development of quantum-communication systems that are usable by non-specialists in quantum physics while also correctly capturing at a useful abstraction the underlying quantum phenomena. Our approach involves defining a quantum abstract machine (QAM) whose operations correspond to well-known quantum circuits; these operations, however, are given direct abstract semantics in a style similar to that of Berry's and Boudol's Chemical Abstract Machine. This paper defines the QAM's semantics and shows via examples how it may be used to model and reason about existing quantum communication protocols.
\end{abstract}

\section{Introduction} \label{sec:introduction}

Quantum computers offer unique capabilities that can be used to
implement substantially faster algorithms than those classical computers are capable of. For example, Grover's search algorithm \cite{grover1996,grover1997} can query unstructured data in sub-linear time (compared to linear time on a classical computer), and Shor's algorithm \cite{shors} can factor integers in polynomial time (no existing classical algorithm for this problem is polynomial-time).
Quantum computing, via mechanisms such as \emph{quantum teleportation}~\cite{PhysRevLett.70.1895,Rigolin_2005},
also supports secure communication capabilities that can transmit information without the possibility of eavesdropping.
These capabilities of quantum computation are motivating researchers to build hybrid classical-quantum communication networks (HCQNs), including so-called \emph{quantum internets}~\cite{Kimble2008},
that integrate quantum computing and existing classical network facilities to provide secure information transmission, with several HCQN protocols being proposed \cite{8068178,https://doi.org/10.48550/arxiv.2205.08479,10.1145/3387514.3405853,e24101488}.

The growing interest in quantum communication systems has also spurred work on formal approaches to modeling and verifying quantum communication protocols.
Frameworks \cite{10.1145/1040305.1040318,10.1145/977091.977108,10.1145/1507244.1507249} based on process algebra have been proposed to model HCQN systems particularly.
These frameworks have generally focused on enhancing existing message-passing process algebras \cite{MILNER19921,DBLP:conf/concur/Sangiori93} with notations for describing \emph{quantum circuits}~\cite{mike-and-ike} and their physical behaviors based on precise linear-algebraic quantum-state interpretations.
The virtue of these approaches is their fidelity to the semantics of quantum states, but with two major issues.
The precise quantum interpretations require classical resources exponential in their quantum descriptions, so the automated verification procedures for such descriptions moreover must confront problems of combinatorial explosion.
Additionally, the message passing model, which the process algebras are based on, can fall afoul of quantum physical realities, such as no relocations or remote controls of quantum qubit resources, so defining protocols in these frameworks requires a significant background in quantum physics to employ;
otherwise, users might define semantically sound protocols that are physically impossible to construct. This fact limits their use to system modelers well-versed in quantum physics.

To address these issues, we introduce the \emph{Quantum Abstract Machine} (QAM),
inspired by the classical \emph{Chemical Abstract Machine}~\cite{BERRY1992217} (CHAM) and Linda~\cite{1663305},
as a basis for the abstract, yet accurate, modeling of quantum communication.
Our approach involves identifying a set of action primitives for modeling HCQN behavior and then giving them direct operational semantics rather than interpreting them in the traditional framework as unitary matrices of complex numbers.
In addition, HCQN behaviors show different functional behaviors from traditional network protocols, which direct the QAM's design.
This yields semantics that abstracts away many low-level details while preserving essential functional HCQN behaviors and physical realities,

While channels in classical networks are identifiers, in HCQNs, quantum channels have a limited lifetime and act as the carriers for conveying quantum messages;
thus, they must be properly created before a quantum message can be delivered, with new operations to connect them with quantum messages in the QAM.
Secondly, the no-cloning theorem \cite{mike-and-ike} indicates that the delivery of a quantum message cannot simply copy the message to the other party,
i.e., once a quantum message is delivered, the sender no longer has the message, which contravenes classical message passing.
Thirdly, as we mentioned above, the physical reality of quantum computing is that quantum qubit resources are localized,
so they cannot be relocated to distinct locations, and there is very limited way, if it is no way, to remotely manipulate quantum resources,
This indicates a regional concept for quantum resources, and we adopt the membrane concept from CHAM to model it.
Finally, HCQNs permit the connections between quantum and classical messages,
e.g., a classical message must be measured out from a quantum qubit in quantum teleportation \cite{PhysRevLett.70.1895},
which involves the nondeterministic generation of values. However, the nondeterminism is not completely random.
Instead, in quantum teleportation, the result of a quantum message measurement creates a pair of classical and quantum residues, and they perform like a key lock pair,
so we can recover the original quantum message by combining the right two residues, as if we match the right key and lock.
All these aspects and the QAM's design are discussed throughout the paper, and the following is our list of contributions.

\begin{itemize}
\item We develop the QAM, with its syntax and semantics (\Cref{sec:sem}), suitable for defining HCQN protocols. To the best of our knowledge, this is the first abstract semantics for describing HCQN behaviors without the involvement of quantum circuits.

\item We carefully employ HCQN properties and physical limitations in the QAM's design so the definable QAM protocols correctly respect these properties and limitations. The QAM helps explore many HCQN properties, such as \Cref{def:no-cloning-good} and \Cref{def:decoding-good}.
  
\item We extend the QAM to describe and evaluate real-world multi-location HCQN protocols (\Cref{sec:case-study}), demonstrating a certain probability of failure. To demonstrate the extendability and utility of the QAM, we utilize trace-refinement in the evaluation framework to reason about the similarity of two HCQN protocols, such as QPASS and QCAST \cite{10.1145/3387514.3405853}.

\end{itemize}

\was{
  were quantum imitations of process
  algebra~\cite{baeten1991process,milner1989communication}, and obtained by
  enhancing a classical process-algebra framework with quantum-circuit
  operations whose semantics is based on interpreting the process-circuit operations over quantum states.

  However, prior frameworks amalgamating linear algebraic quantum circuit
  semantics and process algebra are impractical and unintuitive when it
  comes to verification.
In the classical computing scenario, process algebra suits well with
  automata-based model checking because program semantics can be encoded as state
  transitions.
While in the quantum setting, the verification frameworks become hard to
  define with the presence of matrices for quantum state interpretation.
\mz{I'm not sure if this is convincing, is there any other approach?}
Additionally, those frameworks require the user to assemble communication
  primitives with quantum circuit primitives.
This amounts to defining a full-fledged quantum program that digresses users from
  the original course: defining a communication protocol.
\mz{I'm not familiar with this topic, not sure how to make it convincing.}

  There are two major issues associated with the previous frameworks.
First, the beauty of process algebra is to define a mathematical model that
  captures the essence of multi-threaded program semantics; so that program executions are easily viewed as automata transitions -- usually associated with
  automata-based model-checking mechanisms.
Incorporating quantum-circuit languages significantly complicates the program
  semantics of the previous frameworks. Thus, the associated verification
  framework is unnecessarily complicated.
Second, quantum circuit semantics is unintuitive, and incorporating quantum
  semantics with process algebra makes the system even more unintuitive.
Eventually, programmers, who might have a brief idea about quantum computation,
  need to use a given protocol framework to define HCQN network protocols.
If they spent most of the time working out how the unintuitive framework works,
  why will they use the framework?
Unfortunately, previous HCQN frameworks had unintuitive program semantics
  because they put quantum circuit languages together with
  multi-threaded process algebra together with no chemistry.
}

\ignore{
\liyi{move the beginning of section 3 to intro.}
First, there are two main tasks in sending quantum messages, message transmissions and deliveries.
The former uses quantum swaps to convey a message from one node to another which is closer to the final destination and connectivity is usually represented as a graph structure such as the one in \Cref{fig:q-pi-example}.
During the procedure, the actual quantum swap circuit detail is less interesting than the cost analysis, where each two-node message transmission costs each node one qubit resource, because how swaps are constructed in the circuit is almost identical in different protocols.
We also use a quantum teleportation circuit for delivering a message. Again, the more interesting analysis is about different guarantees a network protocol can provide in the message deliveries rather than the circuit detail.

Next, many quantum message-sending guarantees are time-sensitive and related to probabilities that are time-sensitive.
Thus, we model in QAM a global clock and each task step, a message transmission or delivery, costs a one-time slot,
and every message is associated with a probability value defining the success likelihood of delivering the message.
For example, sending a message from \cn{Cat} to \cn{Dan} in \Cref{fig:q-pi-example} needs an intermediate transmission node $r_1$, and the likelihood of successfully delivering the message is lower than the case if we can directly send the message from \cn{Cat} to \cn{Dan} without an intermediate step. In addition, qubits might be decohered after a certain period of time, so a property we can define,
based on the QAM system, is to guarantee that every delivered message is within a threshold period.
}
 
\section{Overview}
\label{sec:overview}

Many previous quantum process calculi, based on the traditional message passing models,
are incapable of describing HCQN behaviors precisely, e.g., some might be contrary to physical realities,
so programmers might not know if the protocols defined in these systems are physically implementable.
We provide an overview of the QAM designs to unveil these HCQN properties that are hard to capture by message-passing models and show the QAM's design that better represents these properties.
We begin with an example (\Cref{sec:introduction-ex}) showing the abstraction in the QAM.
Essentially, QAM refers to programs as a QAM \emph{configuration} ($\overline{P}$/$\overline{Q}$ in \Cref{fig:q-pi-syntax}) \footnote{In this paper, if \(S\) is a set, then we use $\overline{S}$ to denote the set of multisets over $S$.}, a multiset of membranes; and a program transition in the QAM is a labeled transition system $\overline{P} \xrightarrow{\kappa} \overline{Q}$, with label $\kappa$ introduced in \Cref{fig:q-pi-definition}.
(Background about quantum computing and HCQNs are in \Cref{sec:background}.)

\vspace*{-0.8em}
\subsection{Motivating Example and the QAM Data Design} \label{sec:introduction-ex}
\vspace*{-0.3em}

\begin{figure}[t]
\vspace*{-1.5em}
{\centering
\begin{minipage}[b]{.33\textwidth}
                 \includegraphics[width=1.06\textwidth]{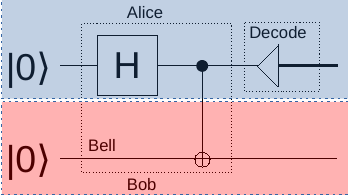}
            \caption{Quantum Channel Circuit}
            \label{fig:channel-circuit-example}
 \end{minipage}
\hfill{}
\hfill{}
\begin{minipage}[b]{.36\textwidth}
                 \includegraphics[width=0.97\textwidth]{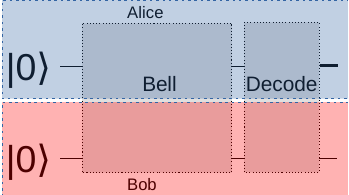}
            \caption{Quantum Channel Task Block}
            \label{fig:channel-block-example}
 \end{minipage}
\hfill{}
\begin{minipage}[b]{.25\textwidth}
                 \includegraphics[width=1\textwidth]{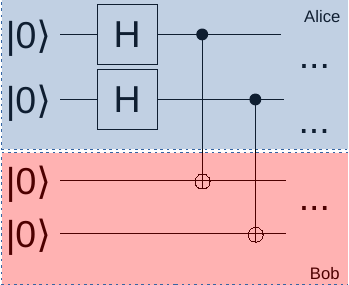}
            \caption{Multi Qubit Quantum Channel}
            \label{fig:multi-channel}
 \end{minipage}
}
\vspace*{-1.5em}
\end{figure}

The primary QAM design goal is to provide a high-level description of HCQN behaviors by abstracting away many quantum circuit-level details based on unitary and density matrix semantics.
To do so, we model the QAM operations based on the task-block diagram concept, which describes many quantum algorithms \cite{mike-and-ike}, to provide a high-level perspective of HCQN protocols.
In the motivated bit-commitment protocol \cite{10.1145/1040305.1040318} (circuit diagram in \Cref{fig:channel-circuit-example}),
Alice and Bob first collaboratively create a quantum channel (Bell pair) connecting them, and then Alice decodes the channel by measuring it.
This results in a classical and quantum residue for Alice and Bob, respectively.
The region for each party is distinguished by color: \textcolor{blue}{blue} for Alice and \textcolor{red}{red} for Bob.
Its task-block diagram (\Cref{fig:channel-block-example}) contains two tasks that form two task-blocks--it connects Alice and Bob through a channel creation followed by a decoding procedure,
which results in the two residues above.

In the QAM, we try to find a succinct set of functional common patterns, such as the Bell pair for channel creation, and abstract away the circuit-level details by utilizing these common patterns, constantly appearing as task blocks in previous works, as our operations.

\vspace*{-0.7em}
\begin{example}[Bit-Commitment]\label{def:exampleover}\rm
The following shows the QAM's implementation of the bit-commitment protocol and one-step transition.

\vspace*{-0.4em}
{\small
\begin{center}
$\pard{\bscell{\pard{\seq{\downa{c}}{\seq{\decode{c}{x}}{0}}}{\circ}}}
{\bscell{\pard{\seq{\downa{c}}{\seq{\creva{c}{y}}{0}}}{\circ}}}
\quad\xrightarrow{c}\quad
\pard{\bscell{\pard{{\seq{\decode{c}{x}}{0}}}{c.\circ}}}
{\bscell{\pard{{\seq{\creva{c}{y}}{0}}}{c.\circ}}}
$
\end{center}
}
\end{example}
\vspace*{-0.4em}

On the left configuration, Alice and Bob hold the left and right membranes \footnote{Membranes are from the CHAM, a group of molecules including processes and resources in $\scell{-}$.}, respectively.
They collaboratively create a quantum channel, named $c$, through the quantum channel creation operation ($\downa{c}$).
During the process, a pair of blank qubit resources $\circ$ are consumed and transferred to the ones marked with the channel name $c$ as $c.\circ$ on the right configuration,
meaning that a channel $c$, with blank content, is established between Alice and Bob.
After that, Alice decodes the channel ($\decode{c}{x}$) and receives classical bits as her residue ($x$),
while Bob waits ($\creva{c}{y}$) for Alice's decoding and receives a quantum residue message.
Note that the QAM abstracts away qubit sizes, e.g., the channel creations in \Cref{def:exampleover} do not necessarily refer to a Bell pair involving only two qubits; they might refer to the multi-qubit channel in \Cref{fig:multi-channel} or a variation of multi-qubit Bell pair admitting circuit optimizations.
This paper refers to the qubit resources as quantum resources, abstracting away qubit sizes.

\begin{figure}[t]
\vspace*{-1.8em}
{\small
  \[\hspace*{-0.3em}
  \begin{array}{l}
\textcolor{blue}{\text{Syntactic Categories:}}\\[0.5em]
\begin{array}{l@{\quad}r@{\;\;}c@{\;\;}l@{\qquad}l@{\quad}r@{\;\;}c@{\;\;}l}
      \text{Variable}& x,y & \in & \mathbb{V} &
      \text{Message Names} & e & \in & \mathbb{E}\\
    \text{Quantum Channel} &\; c,d &\in& \Ls&
    \text{Classical Channel} & a &\in& \mathbb{A}\\ 
    \text{Projective Channel} & \cmsg{c},\cmsg{d} &\in&\cmsg{\Ls}
\end{array}\\[1.5em]
\textcolor{blue}{\text{Message Definitions:}}\\[0.5em]
    \begin{array}{l@{\;\;} r@{\;\;} c@{\;\;} l@{\quad}l@{\;\;} r@{\;\;} c@{\;\;} l} 
\text{Quantum Message} & q &::=& \circ \mid e \mid c.\mu \mid \cmsg{c}.q \mid \up{q}&
\text{Classical Message} & \iota &::=& \cmsg{q} \mid a.\iota \mid \cmsg{c}.\iota\mid \up{\iota}
    \end{array}\\[0.5em]
\textcolor{blue}{\text{Syntactic Sugars:}}\\[0.5em]
    \begin{array}{l@{\quad} r@{\;\;} c@{\;\;} l@{\qquad}l@{\quad} r@{\;\;} c@{\;\;} l} 
\text{Quantum Related Channnels} & \alpha & ::= & c \mid \cmsg{c}&
    \text{All Channels} & \delta & ::= & \alpha\mid a\\
\text{All Messages} & \mu &::=& q \mid \iota&
    \text{Channels \& Messages} & \kappa & ::= & \mu \mid \delta
    \end{array}
    \end{array}
  \]
}
\vspace*{-1.2em}
\caption{QAM Data. $\mathbb{V}$, $\mathbb{E}$, $\Ls$ and $\mathbb{A}$ are disjoint. }
  \label{fig:q-pi-definition}
  \vspace*{-1.5em}
\end{figure}

The QAM's data design is in \Cref{fig:q-pi-definition} based on the hybrid aspect of HCQNs, which connects two worlds: quantum and classical. 
HCQNs hybridize quantum and classical communications based on quantum ($q$) and classical ($\iota$) messages, which are distinctly classified in the QAM (\Cref{fig:q-pi-definition}).
This two-world view is also extended to channels acting as intermediate stations for communications, i.e., we distinguish quantum ($c\in\Ls$) and classical $\Pi$-calculus \cite{MILNER19921} style ($a\in \mathbb{A}$) channels. Moreover, it is necessary to include \emph{projective channels} ($\cmsg{c}\in \cmsg{\Ls}$) as the third kind to transform information between the quantum and classical worlds, explained shortly in \Cref{sec:design-project}.
The QAM includes a singleton quantum message abstraction $e$ and allows $\Pi$-calculus style channeled message $\delta.\kappa$, i.e., the channel $\delta$ contains a content $\kappa$, such as the $c.\circ$ above. A channeled message is either quantum or classical depending on if it contains any quantum pieces as the classification in \Cref{fig:q-pi-definition}.
Getting information from a quantum channel $c$ is a projection from the quantum to the classical world, e.g., Alice's decoding above.
The decoding results two different residue messages: classical $\cmsg{q}$ and quantum residues $\up{q}$.
They are associated with $c$'s projective channel $\cmsg{c}$, as $\cmsg{c}.\cmsg{q}$ and $\cmsg{c}.\up{q}$.
We'll unveil the different QAM data step by step below.

\vspace*{-0.7em}
\subsection{Quantum Resources/Channels are Different from Classical Ones}\label{sec:design-resource}
\vspace*{-0.2em}

As we mentioned in \Cref{sec:introduction}, quantum resources cannot be relocated; remote controls of distinct quantum resources are far from general and limited to a few patterns, e.g.,
The rate of building a quantum channel connecting three remote locations is low \cite{de_Jong_2023}.
A Bell pair or a variation of Bell pair \footnote{A multi-qubit Bell Pair or GHZ that communicates two parties.} that communicates two parties might be the only workable quantum channel to connect two parties. Since every quantum channel in the QAM is between two parties, we refer to one portion of a quantum channel $c$ to be a \emph{party} of $c$, e.g., Alice and Bob's membranes in \Cref{def:exampleover} both hold a party of $c$ after the channel creation.
Three unique features in quantum channels are different from classical channels.

First, the non-relocation property above indicates that we should enforce locality in the QAM, captured by our membrane concept inspired by CHAM.
In \Cref{def:exampleover}, Alice and Bob are in different membranes. The quantum resources ($\circ$ and $c.\circ$) cannot move out of membrane,
even if the multiset structure in each membrane indicates commutativity and associativity.
In addition, Bob's membrane (the right one) cannot contain operations that directly manipulate Alice's quantum resources.

Second, processes and quantum resources in a membrane are separated into different entities to liberate the process operations while guaranteeing the quantum no-cloning theorem.
Many basic programming operations, such as substitutions, are too easily misused, so a no-cloning violation could easily happen if the process-resource separation is not provided.
For example, classical substitutions in message-passing models can easily create multiple copies of a variable in a process.
If such a variable represents a quantum resource, it would violate the no-cloning theorem.
Previous quantum process algebras \cite{10.1145/1040305.1040318,10.1145/977091.977108,10.1145/1507244.1507249} employ additional type and dynamic checks for ensuring no violations.

Inspired by Linda \cite{1663305}, the QAM separates processes and quantum resources and utilizes quantum channels referring to the resources in processes.
In the one-step transition in \Cref{def:exampleover}, after applying the quantum channel creation, we create a channel $c$ connecting the two membranes, e.g., $c.\circ$ in both membranes.
The processes in the two membranes could have other operations arbitrarily copying the quantum channel name $c$. Still, creating multiple copies of the channel name $c$
does not create copies of the quantum resources separated from the processes; thus, we forbid the possibility of cloning the quantum resource labeled as $c$.
In addition, the process resource separation mechanism helps the guarantee of the non-relocation property, e.g.,
Alice (\Cref{def:exampleover}) could missend her quantum channel name $c$ in the process to a third party (Mike) through a classical message-sending operation;
however, Mike's membrane does not contain a quantum resource $c$, located in Alice and Bob's membrane; the classical message exchange about $c$ does not permit Mike's reference to the quantum resource in Alice's membrane.

Third, quantum channels have a lifetime, as they must be created and disappear after being used to communicate a message.
In the QAM, we model the channel creation through a pair of actions $\downa{c}$ in two different membranes, as the \rulelab{Cohere} rule below.

\vspace*{-1em}
{\small
\begin{mathpar}
   \inferrule[Cohere]{}
       { \pard{\scell{\pard{\pard{\seq{\downa{c}}{R}}{\circ}}{\overline{M}}}}{\scell{\pard{\pard{\seq{\downa{c}}{T}}{\circ}}{\overline{N}}}}
        \xrightarrow{c} \pard{\scell{\pard{\pard{{R}}{c.\circ}}{\overline{M}}}}{\scell{\pard{\pard{T}{c.\circ}}{\overline{N}}}}}
  \end{mathpar}
}
\vspace*{-1em}

Both membranes have a process starting with an action ($\downa{c}$), as they intend to create a quantum channel $c$. 
In the creation, each of the two membranes must contain a blank quantum resource $\circ$ for consumption because of the no-cloning theorem, as the reason given in \Cref{sec:background}.
The transition raises a label $c$ and labels the two blank quantum resources with the quantum channel name $c$, indicating that a quantum channel is created.
The created quantum channel can only be used to teleport one message.
A process might indirectly manipulate a quantum resource in the other membrane through a procedure named quantum entanglement swaps, explained in \Cref{sec:enswaps},
where a process teleports the information of a quantum resource to another membrane to create the illusion of prolonging a quantum channel to permit long-distance communications.

\ignore{
\begin{example}[Quantum Teleportation Channel Creation Steps]\label{def:example2}\rm

{\small
\[
\begin{array}{ll}
&
\pard{\bscell{\pard{\seq{\downa{c}}{\seq{ \encode{c}{\cmsg{d}}}{\seq{ \decode{c}{x}}{\seq{{\csenda{a}{x}}}{0}}}}}{\pard{\cmsg{d}.e}{\circ}}}}
{{\bscell{\pard{\seq{\downa{c}}{\seq{ \creva{c}{u}}{\seq{\creva{a}{z}}{\seq{\encode{u}{z}}{0} }}}}{\circ}}}}\\[0.3em]
\equiv
&
\pard{\bscell{\pard{\seq{\downa{c}}{\seq{ \encode{c}{\cmsg{d}}}{\seq{ \decode{c}{x}}{\seq{{\csenda{a}{x}}}{0}}}}}{\pard{\circ}{\cmsg{d}.e}}}}
{{\bscell{\pard{\seq{\downa{c}}{\seq{ \creva{c}{u}}{\seq{\creva{a}{z}}{\seq{\encode{u}{z}}{0} }}}}{\circ}}}}\\[0.3em]
\xrightarrow{c}
&
\pard{\bscell{\pard{\pard{\seq{ \encode{c}{\cmsg{d}}}{\seq{ \decode{c}{x}}{\seq{{\csenda{a}{x}}}{0}}}}{c.\circ}}{{\cmsg{d}.e}}}}
{\bscell{{\pard{\seq{ \creva{c}{u}}{\seq{\creva{a}{z}}{\seq{\encode{u}{z}}{0} }}}{c.\circ}}}}
\end{array}
\]
}

\end{example}

The above example is the first two transitions for quantum teleportation (\Cref{def:example1}).
Here, we first equationally rewrite the elements in the first membrane (Alice) so that the blank quantum resource $\circ$ is adjacent to Alice's process.
The rewrite is validated by the commutativity equational property of QAM membranes.
Then, we apply rule \rulelab{Cohere} to create a fresh channel $c$ that permits the communication of the two parties.
}

\vspace*{-0.7em}
\subsection{Projective Nondeterminism and Projective Channels}\label{sec:design-project}
\vspace*{-0.3em}

The decoding operation result of bit-commitment (\Cref{def:exampleover}) nondeterministically produces a pair of classical and quantum residues, selected from the two pairs: $0$ and $\ket{0}$ as well as $1$ and $\ket{1}$. The non-determinism is not truly random, in the sense that one can recover Alice's qubit $\ket{0}$ by merging any of the two pairs.
The same happens in quantum teleportation in \Cref{def:example1}, i.e., any nondeterministically generated pair after the decoding can be merged to recover Alice's original quantum message. An analogy of the phenomenon is that a decoding operation nondeterministically creates a pair of keys and locks, so any pair is matched, but crossed pair keys and locks fail to match.
We name this kind of nondeterminism to be \emph{projective non-determinism} because the nondeterminism is caused by the fact that we project the information in the quantum world to the classical world through the decoding operation.
To model the phenomenon, the QAM includes the concept of projective channels ($\cmsg{c}$ in \Cref{fig:q-pi-definition}). After decoding a quantum channel, we label the two residues in a newly generated projective channel, indicating that they are strongly correlated and can be merged to recover the original information in the quantum channel. Details are in \Cref{sec:overviewprotocol}.

\ignore{
The biggest feature of the QAM is to abstract away the low-level circuit details of HCQN protocols, so one can define these protocols in a high-level abstraction through the view of task blocks as quantum operations.
The abstraction has two requirements: 1) it needs to respect quantum computation theories and physical realities of HCQN communications, and 2) the configurations should be intuitive enough for users to understand.
We develop the QAM by migrating the concepts from previous process algebra (CHAM and Linda), with modifications,
to reflect and respect HCQN communications.

In classical process algebra, such as $\Pi$-calculus \cite{MILNER19921},
the processes and messages are strongly tied together, and process semantics are defined through variable/message substitutions.
After such a substitution, there might exist multiple copies of a message,
which would unavoidably cause the algebra to violate the no-cloning theorem if we simply migrated the variable/message substitution concept to develop a quantum process algebra.
We develop the QAM by including ideas from Linda to separate processes and quantum information resources and parallelize them as molecules in a membrane.
Only when a process and a quantum resource are adjacent, inspired by CHAM, can we allow the process to manipulate the resource. For example, in the bit-commitment configuration above, to create the channel $c$, we require the blank quantum resource $\circ$ to be adjacent to the process containing the channel creation operation in each membrane, shown in the left configuration above. Then, the two $\downa{c}$ operations in the processes can manipulate the two blank resources to create a quantum channel.

In addition, we also realize that HCQN communications typically have distance restrictions, such that a process can react to talk to a quantum resource if they are close to each other. In the QAM, we restrict the above reaction only if molecules (processes and resources) are from the same membrane. Cross-membrane communications can only happen through the bridges of channels.

Note that the QAM abstracts away qubit sizes, so the qubit resources in the QAM are parameterized by a qubit bandwidth when we compile the QAM to an HCQN circuit.
For example, the above quantum channel creation in the bit-commitment configuration does not necessarily refer to the Bell pair creation involving only two qubits.
If the pre-defined qubit size is sufficiently large, the creation might refer to the multi-qubit channel in \Cref{fig:multi-channel}.
This paper refers to the qubit resources as quantum resources, meaning they are not restricted to singleton qubits.
}

\vspace*{-0.8em}
\section{The QAM Semantics}
\label{sec:sem}
\vspace*{-0.3em}

In this section, we define the QAM semantics based on a succinct set of abstract operations to capture the behaviors of HCQN protocols.
For highlighting the operations, we utilize quantum teleportation as an example,
which is the standard HCQN protocol, serving as the basis of many other HCQN protocols \cite{10.1145/3386367.3431293,10.1145/3387514.3405853}.

\vspace*{-0.7em}
\subsection{The QAM Operation Design for Quantum Communications}\label{sec:qamsyntax}
\vspace*{-0.2em}

\begin{figure}[t]
\vspace*{-2.3em}
{\small
  \[\hspace*{-1em}
    \begin{array}{l@{\qquad} l@{\;\;} c@{\;\;\;} l} 
    \text{Resources} & \phi & ::= & \alpha.\mu \mid \circ\\
      \text{Action} & A & ::= & \downa{c} \mid \csenda{a}{\iota}\mid\creva{\delta}{x}
\mid \encode{\alpha}{\kappa} \mid \decode{c}{x}\\[0.2em]

      \text{Process} & R,T & ::= & 0 \mid \seq{A}{R}\mid \parp{R}{T}\mid \comp{R} \\[0.2em]

      \text{Molecule} & M,N & ::= & R \mid \phi\\[0.2em]

      \text{Membrane} & P,Q &::=& \scell{\overline{M}}
                         \mid \parll{P}{\phi}{Q}
    \end{array}
  \]
}
\vspace*{-1.5em}
\caption{QAM syntax. $\overline{S}$ to denote the set of multisets over $S$. }
  \label{fig:q-pi-syntax}
  \vspace*{-1.5em}
\end{figure}

The QAM syntax is given in \Cref{fig:q-pi-syntax}.
The CHAM's membrane concept \cite{BERRY1992217} is central in the QAM's operation design, where processes and quantum resources coexist as molecules in a membrane, inspired by Linda \cite{1663305}. In a membrane, we utilize the multiset adjacency mechanism from CHAM to marry a process and a quantum resource, permitting the manipulation of the resource. On the other hand, we only allow communications through channels between distinct membranes.
We have two kinds of membranes:
a molecule membrane ($\scell{\overline{M}}$) indicates that all molecules $\overline{M}$, either processes or resources, can marry with each other;
and an airlocked membrane $\parll{P}{\phi}{Q}$, i.e., the membranes $P$ and $Q$ are actively communicating through the shared resource $\phi$, typically a quantum channel.
Every process living inside membranes is one of the following actions: a unit $0$, an action followed by a process $\seq{A}{R}$, a choice operation $\parp{R}{T}$, or a process replication $\comp{R}$.

We model a concise set of operations --- mainly inspired by quantum task blocks repeatedly appearing in quantum network algorithms --- as actions,
including similar operations from $\Pi$-calculus such as classical message sender ($\csenda{a}{\mu}$) and receiver ($\creva{\delta}{x}$), as well as new QAM operations for modeling HCQN protocols, such as quantum channel creators ($\downa{c}$), quantum message encoders ($\encode{c}{\kappa}$) and decoders ($\decode{c}{x}$).

\vspace*{-0.5em}
\begin{example}[Quantum Teleportation]\label{def:example1}\rm
The QAM configuration for quantum teleportation.

\vspace*{-0.6em}
{\small
\begin{center}
$
\begin{array}{ll}
&
\pard{\bscell{\pard{\seq{\downa{c}}{\seq{ \encode{c}{\cmsg{d}}}{\seq{ \decode{c}{x}}{\seq{{\csenda{a}{x}}}{0}}}}}{\pard{\cmsg{d}.e}{\circ}}}}
{{\bscell{\pard{\seq{\downa{c}}{\seq{ \creva{c}{u}}{\seq{\creva{a}{z}}{\seq{\encode{u}{z}}{0} }}}}{\circ}}}}\\
\xrightarrow{c}
&
\pard{\bscell{\pard{\pard{\seq{ \encode{c}{\cmsg{d}}}{\seq{ \decode{c}{x}}{\seq{{\csenda{a}{x}}}{0}}}}{c.\circ}}{\cmsg{d}.e}}}
{\bscell{\pard{\seq{ \creva{c}{u}}{\seq{\creva{a}{z}}{\seq{\encode{u}{z}}{0} }}}{c.\circ}}}
\end{array}
$
\end{center}
}
\vspace*{-0.8em}
\end{example}

The configuration is similar to \Cref{def:example1}. Alice has an additional encoding that encodes a quantum message resource $\cmsg{d}.e$, in her membrane, waiting to send to Bob.
Alice also sends the classical residue after her decoding to Bob via a classical channel $a$.

\vspace*{-0.9em}
\subsubsection{Equations for Membranes}
\label{sec:equations}
\vspace*{-0.3em}

The QAM's multiset structures indicate implicit equational properties, i.e., identity, commutativity, and associativity equational properties.
Another important equational property is the airlock property (\rulelab{AirEQ}), an imitation of the airlock in CHAM \footnote{CHAM's airlock happens inside a membrane, while ours happens across membranes.}, where we airlock two distinct membranes sharing a quantum channel for activating their communication.

{
{\small
\begin{mathpar}
   \inferrule[AirEQ]{}
       {\pard{\scell{\pard{{c.\mu_1}}{\overline{M}}}}{\scell{\pard{{c.\mu_2}}{\overline{N}}}}
             \equiv 
             \parll{\scell{{\overline{M}}}}{c.(\mu_1\odot\mu_2)}{\scell{{\overline{N}}}}}
\end{mathpar}
}
\vspace*{-1.5em}
}

The equation \rulelab{AirEQ} describes the above communication activation. On the right side, the two membranes hold the same channel $c$ with different content.
After the equational rewrite, the two membranes actively communicate through the channel $c$ with the combination of the two content $\mu_1 \odot \mu_2$.
The $\odot$ operation describes an algebra for calculating the composition of two quantum messages (\Cref{sec:qamsem}),
which is extendable, inspired by the extendable features in the CHAM, i.e.,
We can extend the operation to include circuit gates for reasoning about complicated quantum configurations written in circuit gates as a future work of the QAM, e.g., the one in \Cref{sec:qamsecurity}.
Equations are bidirectional so that an airlock mode can be dissolved to the original membrane as the left-hand side of the equation \rulelab{AirEQ}.

\vspace*{-0.8em}
\subsubsection{Encoding Quantum Messages into Channels}
\label{sec:quantummessage}
\vspace*{-0.3em}

Message encoding is the procedure of installing a message, either quantum or classical, to a quantum channel, which is one of the essential steps in HCQN communications,
e.g., in \Cref{def:example1}, Alice encodes a quantum message into a channel between Alice and Bob.
Different encoding schemes are modeled as rules below.

\vspace*{-1.2em}
{\small
\begin{mathpar}
   \inferrule[QLocal]{}
       {\scell{\pard{\pard{\seq{\encode{c}{\alpha}{R}}}{\alpha.q}}{...}}
             \longrightarrow \scell{\pard{\seq{\encode{c}{\alpha.q}}{R}}{...}}}

   \inferrule[CLocal]{}
       {\scell{\pard{\pard{\seq{\creva{\cmsg{c}}{x}}{R}}{\cmsg{c}.\iota}}{{...}}}
             \longrightarrow {\scell{\pard{{R[\iota/x]}}{...}}} }
  \end{mathpar}
  \begin{mathpar}
   \inferrule[Encode]{}
       {\scell{\pard{\pard{\seq{\encode{\alpha}{\mu_1}}{R}}{\alpha.\mu_2}}{...}}
             \longrightarrow {\scell{\pard{\pard{R}{\alpha.(\mu_1\odot \mu_2)}}{...}} }}
  \end{mathpar}
}
\vspace*{-1.2em}

In rule \rulelab{Encode}, process $R$ is adjacent to a party of a quantum channel $c$ with the content $\mu_2$. $R$ applies an operation to encode the message $\mu_1$ to $c$, turning $c$'s content to be $\mu_1\odot \mu_2$ by summing $\mu_1$ and $\mu_2$ through the meet operation $\odot$; explained shortly below.

The other two rules help the localization of quantum resources.
In QAM, a quantum resource typically does not appear in a process due to the separation of resource and process. 
However, we need to localize a quantum resource in rule \rulelab{Encode}.
The localization is handled by rules \rulelab{QLocal} and \rulelab{CLocal}.
The former localizes the quantum resource, while the latter handles the one with a classical message $\iota$ having a projective channel $\cmsg{c}$.
When a process encodes a quantum message $\alpha.q$, the encoding operation (${c}\triangleleft{\alpha}$) in the process first refers the channel $\alpha$ to the resource location where the message resides. Then, the process localizes the message by calling the referenced name $\alpha$, as in rule \rulelab{QLocal},
where we replace the channel $\alpha$ in $R$ with the message $\alpha.q$.
Here, we only replace a single channel $\alpha$ with the single quantum message without violating the no-cloning theorem -- a quantum message cannot be cloned.
Rule \rulelab{CLocal} does a similar task as \rulelab{QLocal}, but we can replace every occurrence of $x$ with $\iota$ in $R$ because $\iota$ is classical without the no-cloning limitation.
We will explain why a resource molecule can be a classical message with the form $\cmsg{d}.q$ in \Cref{sec:overviewprotocol}.
We proved the property below that the QAM semantics does not violate the no-cloning theorem.

\vspace*{-0.5em}
\begin{theorem}[No-cloning Assurance]\label{def:no-cloning-good}\rm

For any membrane $\scell{\pard{\overline{R}}{\overline{\phi}}}$ in a QAM configuration $\overline{P}$, for any quantum message $q$ as a resource molecule in $\overline{\phi}$, any transition $\overline{P} \to \overline{Q}$ does not substitute and copy $q$ more than once.
\end{theorem}
\vspace*{-0.5em}

A proof outline for the above theorem is given in \Cref{sec:qamsemproof1a}.

\begin{example}[The Quantum Teleportation Encoding Steps]\label{def:example3}\rm

{\small
\[\vspace*{-0.7em}
\begin{array}{lll}
&
\pard{\bscell{\pard{\pard{\seq{ \encode{c}{\cmsg{d}}}{\seq{ \decode{c}{x}}{\seq{{\csenda{a}{x}}}{0}}}}{c.\circ}}{\cmsg{d}.e}}}
{\bscell{\pard{\seq{ \creva{c}{u}}{\seq{\creva{a}{z}}{\seq{\encode{u}{z}}{0} }}}{c.\circ}}}\\[0.2em]
\equiv
&
\pard{\bscell{\pard{\pard{\seq{ \encode{c}{\cmsg{d}}}{\seq{ \decode{c}{x}}{\seq{{\csenda{a}{x}}}{0}}}}{\cmsg{d}.e}}{c.\circ}}}
{\bscell{\pard{\seq{ \creva{c}{u}}{\seq{\creva{a}{z}}{\seq{\encode{u}{z}}{0} }}}{c.\circ}}}\\[0.2em]
\longrightarrow
&
\pard{\bscell{\pard{\seq{ \encode{c}{\cmsg{d}.e}}{\seq{ \decode{c}{x}}{\seq{{\csenda{a}{x}}}{0}}}}{c.\circ}}}
{\bscell{\pard{\seq{ \creva{c}{u}}{\seq{\creva{a}{z}}{\seq{\encode{u}{z}}{0} }}}{c.\circ}}}
&
(\rulelab{QLocal})
\\[0.2em]
\longrightarrow
&
\pard{\bscell{\pard{\seq{ \decode{c}{x}}{\seq{{\csenda{a}{x}}}{0}}}{c.\cmsg{d}.e}}}
{\bscell{\pard{\seq{ \creva{c}{u}}{\seq{\creva{a}{z}}{\seq{\encode{u}{z}}{0} }}}{c.\circ}}}
&
(\rulelab{Encode})
\end{array}
\]
}
\end{example}

\Cref{def:example3} shows the encoding steps in \Cref{def:example1} after creating a quantum channel.
Here, the localization rule \rulelab{QLocal}, after the equational rewrite, replaces the spot marked as  $\cmsg{d}$ with the message $\cmsg{d}.e$.
Then, we apply rule \rulelab{Encode} to push $\cmsg{d}.e$ to the quantum channel $c$.
Bob's party of the channel $c$ is unaffected in the encoding transition.
This non-affection unveils that Bob might not observe Alice's encoding behavior, coinciding with quantum channels' physical behavior.
A channel's content depends on the meet ($\odot$) of the two pieces of information in the channel's two parties, described in rule \rulelab{AirEQ} above.
The meet operation $\odot$ describes the behavior of quantum message manipulations from two different membranes.
In QAM, we first propose its basic algebraic property as:

\vspace*{-0.5em}
\begin{definition}[Message Encoding Properties]\label{def:theoremassem}\rm
The $\odot$'s algebraic properties:

\vspace*{-0.5em}
{\small
\begin{center}
$
\mu \odot \mu' \equiv \mu' \odot \mu
\quad
\circ \odot \,q \equiv q
\quad
\cmsg{c}.\cmsg{q}\; \odot \cmsg{c}.\up{q} \equiv q
\quad
\iota \odot \iota \equiv \circ
\quad
\cmsg{c}.(\cmsg{c}.\iota \odot \mu) \equiv \cmsg{c}.\iota \odot \cmsg{c}.\mu
$
\end{center}
}
\vspace*{-1em}
\end{definition}

The first two rules define $\odot$'s commutativity and identity.
The third rule introduces the meet of classical and quantum residues.
The next suggests that the effect is canceled if two parties encode the same classical message into a quantum channel.
The last suggests that the meet of nested projective channels can be rearranged if they are the same; see \Cref{sec:qamsem}.

\vspace*{-0.7em}
\subsubsection{Decoding a Quantum Channel}
\label{sec:overviewprotocol}
\vspace*{-0.2em}

Once a party encodes a message in a quantum channel, he can apply the decoding operation to cut off his holding of the channel
so that the quantum information stored in the quantum channel flows to another party.
Regarding the channel's content type (classical or quantum), the decoding semantics are described below by the \rulelab{Decode} rule.

\vspace*{-1em}
{\small
\begin{mathpar}
   \inferrule[Decode]{}
       {\parll{\scell{\pard{\seq{\decode{c}{x}}{R}}{\overline{M}}}}{c.\mu}{\scell{\pard{\seq{\creva{c}{y}}{T}}{\overline{N}}}}
                \xrightarrow{c.\mu} 
      \pard{\scell{\pard{R[\cmsg{c}.\cmsg{\mu}/x]}{\overline{M}}}}{\scell{\pard{\pard{T[\cmsg{c}/y]}{\cmsg{c}.\up{\mu}}}{\overline{N}}}} }

  \end{mathpar}
}
\vspace*{-1em}

Before a decoding transition, the left and right membranes need to be airlocked via a shared quantum channel $c.\mu$.
Process $R$ starts decoding the channel $c$, and it receives a classical residue $\cmsg{c}.\cmsg{\mu}$
with the projective channel $\cmsg{c}$. Process $T$ in the right membrane is waiting ($\creva{c}{y}$) on the decoding result via channel $c$, and it receives a projective channel $\cmsg{c}$.
In addition, the right membrane also contains a new residue resource $\cmsg{c}.\up{\mu}$ with the projective channel $\cmsg{c}$.
If the channel content is quantum, i.e., $\mu$ is a quantum message $q$, $\up{q}$ is a quantum residue;
if $\mu$ is classical ($\iota$), $\up{\iota}$ is a classical residue.
Essentially, a projective channel holding a classical message is a classical entity and should not be considered a quantum resource.
One can utilize rule \rulelab{CLocal} to localize the message to be owned by a process.

\vspace*{-0.2em}
\begin{example}[The Quantum Teleportation Decoding Steps]\label{def:example4}\rm

{\small
\[\vspace*{-0.8em}
\begin{array}{ll}
&
\pard{\bscell{\pard{{\seq{ \decode{c}{x}}{\seq{{\csenda{a}{x}}}{0}}}}{c.\cmsg{d}.e}}}
{\bscell{\pard{\seq{ \creva{c}{u}}{\seq{\creva{a}{z}}{\seq{\encode{u}{z}}{0} }}}{c.\circ}}}
\\
\equiv
&
\parll{{\bscell{{{\seq{ \decode{c}{x}}{\seq{{\csenda{a}{x}}}{0}}}}}}}{c.\cmsg{d}.e}
{{\bscell{{\seq{ \creva{c}{u}}{\seq{\creva{a}{z}}{\seq{\encode{u}{z}}{0} }}}}}}
\\
\xrightarrow{c.\cmsg{d}.e}
&
\pard{{\bscell{{{{\seq{{\csenda{a}{\cmsg{c}.\cmsg{\cmsg{d}.e}}}}{0}}}}}}}
{\bscell{\pard{\seq{\creva{a}{z}}{\seq{\encode{\cmsg{c}}{z}}{0} }}{\cmsg{c}.\up{\cmsg{d}.e}}}}
\end{array}
\]
}

\end{example}
\vspace*{-0.3em}

The above shows the decoding steps after \Cref{def:example3},
where we decode the channel $c$ with quantum content $\cmsg{d}.e$, by first airlocking the two membranes.
Decoding $c$ pushes the left membrane to receive a classical residue $\cmsg{c}.\cmsg{\cmsg{d}.e}$ and the right membrane receives the projective channel $\cmsg{c}$.
The right membrane's resource is rearranged as the quantum residue $\cmsg{c}.\up{\cmsg{d}.e}$.
The label $c.\cmsg{d}.e$ refers to a quantum channel $c$ with a quantum message $e$ labeled with a projective channel $\cmsg{d}$.

The following theorem explains the non-relocation property in \Cref{sec:design-resource}, i.e., message decoding is the sole way of teleporting the information stored in a quantum resource $c$ from one membrane $P$ to the other one $Q$. We show that the QAM respects the property below:

\begin{theorem}[Non-Relocation]\label{def:decoding-good}\rm
For any $P=\scell{\pard{\overline{R}}{\overline{\phi}}}$ in a QAM configuration $\overline{P}$, for a quantum resource labeled with $c$ in $\overline{\phi}$,
if it is transformed to another membrane $Q$ in $\overline{P}$, then

\vspace*{-0.5em}

\begin{itemize}

\item There must be a decoding operation to erase $c$ in $P$ first.
\item After the decoding, there is an encoding in $Q$ to reconstruct $c$ in $Q$'s resource fields.

\end{itemize}
\end{theorem}
\vspace*{-0.5em}

A proof outline for the above theorem is given in \Cref{sec:qamsemproof1b}.
The key indication in \Cref{def:decoding-good} shows that classical message exchanging of a quantum channel name $c$ does not permit the other party access to the resource $c$.

\vspace*{-0.3em}
\begin{example}[The Quantum Teleportation Reconstruction Steps]\label{def:example4a}\rm

{\small
\[\vspace*{-0.5em}
\begin{array}{ll}
&
\pard{{\bscell{{{{\seq{{\csenda{a}{\cmsg{c}.\cmsg{\cmsg{d}.e}}}}{0}}}}}}}
{\bscell{\pard{\seq{\creva{a}{z}}{\seq{\encode{\cmsg{c}}{z}}{0} }}{\cmsg{c}.\up{\cmsg{d}.e}}}}
\\
\xrightarrow{a.\cmsg{c}.\cmsg{\cmsg{d}.e}}
&
\pard{{\bscell{{{{{0}}}}}}}
{\bscell{\pard{{\seq{\encode{\cmsg{c}}{\cmsg{c}.\cmsg{\cmsg{d}.e}}}{0} }}{\cmsg{c}.\up{\cmsg{d}.e}}}}
\\
\longrightarrow
&
\pard{{\bscell{{{{{0}}}}}}}
{\bscell{\pard{0}{\cmsg{d}.e}}}
\end{array}
\]
}
\end{example}
\vspace*{-0.3em}

The above example shows the final two steps of quantum teleportation. Alice sends the classical residue through classical channel $a$, and Bob encodes the received message to the quantum resource with the projective channel $\cmsg{c}$ to retrieve $\cmsg{d}.e$.

\vspace*{-0.9em}
\subsection{The QAM Semantic Framework and Other Semantic Rules} \label{sec:qamsem}

\begin{figure}[h]
\vspace*{-1.8em}
{\footnotesize
\begin{mathpar}
\hspace*{-0.8em}
   \inferrule[ID1]{}
       {\scell{\emptyset} \equiv \emptyset   }
\qquad
   \inferrule[ID2]{}
       {0 \equiv \emptyset   }
\qquad
   \inferrule[Split]{}
       {\scell{\pard{\overline{M}}{\overline{N}}}\longrightarrow \pard{\scell{\overline{M}}}{\scell{\overline{N}}}  }
\qquad
  \inferrule[CL]{}
      {\parp{R}{T} \longrightarrow R}
\qquad
  \inferrule[CR]{}
      {\parp{R}{T} \longrightarrow T}
\qquad
   \inferrule[MT]{}
       {\comp{R} \longrightarrow \pard{R}{\comp{R}}}
  \end{mathpar}
  \begin{mathpar}
   \inferrule[NT]{}
       {\comp{R} \longrightarrow 0}

   \inferrule[Decohere]{}
       { \phi \longrightarrow \emptyset }

  \inferrule[Com]{}
      { \pard{\scell{\pard{\seq{\csenda{a}{\iota}}{R}}{...}}}{\scell{\pard{\seq{\creva{a}{x}}{T}}{...}}}
           \xrightarrow{a.\iota} \pard{\scell{\pard{R}{...}}}{\scell{\pard{T[\iota/x]}{...}}}}
  \end{mathpar}
}
\vspace*{-1.8em}
\caption{Additional QAM semantics.}
  \label{fig:q-pi-semantics1}
  \vspace*{-1em}
\end{figure}

\Cref{fig:q-pi-semantics1} provides additional equations and semantic rules.
Rules \rulelab{ID1} and \rulelab{ID2} define the identity relations of membranes and processes.
Rule \rulelab{Split} separates a membrane into two.
In QAM, a membrane can be split into two, but two membranes cannot join into one.
If a membrane wants to create a quantum channel with itself, it can first apply the \rulelab{Split} rule and then apply the channel creation rule.

Rules \rulelab{CL} and \rulelab{CR} define the semantics for choice operations,
while rules \rulelab{MT} and \rulelab{NT} define the semantics for replications, both similar to the ones in $\Pi$-calculus.
The behavior of replications is similar to a local quantum machine that can repeatedly generate the same quantum state,
which does not violate the no-cloning theorem because the generation is local and only produces the same known quantum state.
Rule \rulelab{Decohere} defines the phenomenon that quantum resources decohere which can cause a quantum resource disappear nondeterminically.
We also support $\Pi$-calculus style classical communication as the rule \rulelab{Com}.
In rule \rulelab{Decode} (\Cref{sec:overviewprotocol}) and \rulelab{Com}, the receiver operation $\creva{a}{x}$ acts as synchronous waiting.
In \rulelab{Decode}, it waits for the other party to decode a quantum channel while in \rulelab{Com}, it waits to receive a classical message ;
examples are in \Cref{def:example4,def:example4a}.

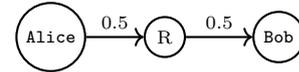
\begin{wrapfigure}{r}{3.8cm}
\vspace*{-3em}
{\scriptsize
    \begin{tikzpicture}[node distance={1.5cm}, thick, main/.style = {draw, circle}] 
    \node[main] (1) {\cn{Alice}}; 
    \node[main] (2) [right of=1] {R};
    \node[main] (3) [right of=2] {\cn{Bob}};
    \draw[->] (1) -- node[midway, above] {0.5} (2); 
    \draw[->] (2) -- node[midway, above] {0.5} (3);  
    \end{tikzpicture}
}
    \caption{Alice to Bob via an Intermediate Router}
    \label{fig:ctdan}
\end{wrapfigure} 

\vspace*{-0.9em}
\subsection{Case Study: Quantum Entanglement Swaps}
\label{sec:enswaps}
\vspace*{-0.3em}

As mentioned in \Cref{sec:introduction-ex}, communication through a two-party quantum channel might be the only available form of long-distance quantum communication.
Essentially, long-distance quantum message transmissions are built through the compositions of multiple quantum teleportations; this is known as quantum entanglement swaps (QES in \Cref{sec:background}).
QES treats a party of a quantum channel $c$ as a piece of quantum information that can be encoded into another channel $d$ and conveyed through $d$.
We show an example QES configuration below for modeling \Cref{fig:background-circuit-examplec}, with the left membrane being Alice, the middle one being an intermediate router, and the right one being Bob. Alice tries to send a message $\cmsg{c}.e$ to Bob.

\begin{example}[Quantum Entanglement Swap]\label{def:example6}\rm
The QES configuration and its first few steps of transitions.
Here, $A={{\seq{ \decode{c}{x}}{\seq{{\csenda{a}{x}}}{0}}}}$, $R=\decode{d}{y}\seq{\csenda{a_1}{y} }{0}$ and $B=\seq{\creva{c}{u}\creva{a}{t}\encode{u}{t}}{0}$.

{\small
$
\pard{\pard{\bscell{\pard{\seq{\downa{c}}{\seq{ \encode{c}{\cmsg{c_1}}}{A}}}{\pard{\cmsg{c_1}.e}{\circ}}}}
{{\bscell{\pard{\seq{\downa{c}\downa{d}\encode{d}{c}}{R}}{\pard{\circ}{\circ}}}}}}
{
{{\bscell{\pard{\seq{\downa{d}}{\seq{ \creva{d}{w}}{\seq{\creva{a_1}{z}}{\seq{\encode{w}{z}}{B }}}}}{\circ}}}}}
$

{\noindent
\[
\begin{array}{lll}
...\xrightarrow{c}...\xrightarrow{d}
&
\pard{\pard{\bscell{\pard{A}{c.\cmsg{c_1}.e}}}
{\bscell{\pard{\pard{\seq{\encode{d}{c}}{R}}{d.\circ}}{{c.\circ}}}}}
{{\bscell{\pard{\seq{ \creva{d}{w}}{\seq{\creva{a_1}{z}}{\seq{\encode{w}{z}}B }}}{d.\circ}}}}
\\[0.2em]
\equiv
&
\pard{\pard{\bscell{\pard{A}{c.\cmsg{c_1}.e}}}
{\bscell{\pard{\pard{\seq{\encode{d}{c}}{R}}{c.\circ}}{{d.\circ}}}}}
{{\bscell{\pard{\seq{ \creva{d}{w}}{\seq{\creva{a_1}{z}}{\seq{\encode{w}{z}}B }}}{d.\circ}}}}
\\[0.2em]
\longrightarrow
&
\pard{
\pard{\bscell{\pard{A}{c.\cmsg{c_1}.e}}}
{{\bscell{\pard{{\seq{\encode{d}{c.\circ}}{R}}}{{{d.\circ}}}}}}}
{
{{\bscell{\pard{\seq{ \creva{d}{w}}{\seq{\creva{a_1}{z}}{\seq{\encode{w}{z}}{B} }}}{d.\circ}}}}}
&{(\rulelab{QLocal})}
\\[0.2em]
\equiv
&
\pard
{\pard{\bscell{\pard{A}{c.\cmsg{c_1}.e}}}
{{\bscell{\pard{\seq{\encode{d}{c.\circ}}{R}}{d.\circ}}}}}
{
{{\bscell{\pard{\seq{ \creva{d}{w}}{\seq{\creva{a_1}{z}}{\seq{\encode{w}{z}}{B} }}}{d.\circ}}}}}
\\[0.2em]
\longrightarrow
&
\pard
{\pard{\bscell{\pard{A}{c.\cmsg{c_1}.e}}}
{{\bscell{\pard{R}{d.c.\circ}}}}}
{
{{\bscell{\pard{\seq{ \creva{d}{w}}{\seq{\creva{a_1}{z}}{\seq{\encode{w}{z}}{B} }}}{d.\circ}}}}}
&{(\rulelab{Encode})}
\end{array}
\]
}
}
\end{example}

The example above shows several transition steps. The first step is an abbreviation of combining several steps to create two quantum channels, $c$ and $d$.
Rule \rulelab{QLocal} replaces the channel name $c$ in the encoding operation $\encode{d}{c}$, with $c.\circ$, to localize a party of the channel $c$, which will be viewed as a piece of quantum information encoded into the channel $d$.
Realize that the other party of the channel $c$ held by Alice does not change during the process, as Alice does not notice the change happening in the other part of the same channel.
The final step encodes the channel party $c.\circ$ to the channel $d$ by treating $c.\circ$ as a quantum message.

After the above transitions, the later ones in QES are similar to the last few steps of the quantum teleportation algorithm, where we decode the quantum channel $d$ to hand over the quantum channel $c$ to Bob. The transitions are described as follows.

\begin{example}[QES Decoding for Long Distance Channel Building]\label{def:example7}\rm

{\small
\[
\begin{array}{ll}
&
\pard
{\pard{\bscell{\pard{A}{c.\cmsg{c_1}.e}}}
{{\bscell{\pard{R}{d.c.\circ}}}}}
{{\bscell{\pard{\seq{ \creva{d}{w}}{\seq{\creva{a_1}{z}}{\seq{\encode{w}{z}}{B} }}}{d.\circ}}}}
\\[0.2em]
\equiv
&
\pard
{\bscell{\pard{A}{c.\cmsg{c_1}.e}}}
{
\parll{{\bscell{\decode{d}{y}\seq{\csenda{a_1}{y} }{0}}}}{d.c.\circ}
{{\bscell{{\seq{ \creva{d}{w}}{\seq{\creva{a_1}{z}}{\seq{\encode{w}{z}}{B }}}}}}}}
\\[0.2em]
\xrightarrow{d.c.\circ}
&
\pard
{\bscell{\pard{A}{c.\cmsg{c_1}.e}}}
{
\pard{{\bscell{{{\seq{\csenda{a_1}{\cmsg{d}.\cmsg{c.\circ}} }{0}}}}}}
{{\bscell{\pard{\seq{\creva{a_1}{z}}{\seq{\encode{\cmsg{d}}{z}}{B} }}{\cmsg{d}.\up{c.\circ}}}}}}
\\[0.2em]
\xrightarrow{a_1.\cmsg{d}.\cmsg{c.\circ}}
&

\pard{\bscell{\pard{A}{c.\cmsg{c_1}.e}}}
{
\pard{{\bscell{0}}}
{{\bscell{\pard{{\seq{\encode{\cmsg{d}}{\cmsg{d}.\cmsg{c.\circ}}}{B} }}{\cmsg{d}.\up{c.\circ}}}}}}
\\[-.2em]
\longrightarrow
&
\pard{\bscell{\pard{A}{c.\cmsg{c_1}.e}}}
{\pard{{\bscell{0}}}
{{\bscell{\pard{B }{c.\circ}}}}}
\end{array}
\]
}
\end{example}

In the above transitions, we aim to teleport $c.\circ$ from the router to Bob so that Bob can communicate with Alice via the channel $c$.
As we can see, the procedures are the same as the second half of the quantum teleportation transitions in \Cref{def:example4},
except that the classical and quantum residues become $\cmsg{d}.\cmsg{c.\circ}$ and $\cmsg{d}.\up{c.\circ}$, respectively.
After the final recovery encoding step, the channel $d$ disappears, and Bob holds the channel $c.\circ$.
The step essentially transforms the quantum resource in $d$ to be the quantum information of the channel $c.\circ$ by treating $c.\circ$ as a piece of quantum information.
Then, Alice and Bob can communicate through channel $c$ by executing the second half of the quantum teleportation transitions (\Cref{def:example4}) again. In the end, Bob receives the quantum message $\cmsg{c_1}.e$.

QES is an important protocol for long-distance transmissions in HCQN communications.
We will see an evaluation framework can be built on top of the QAM in \Cref{sec:case-study}.

\ignore{
including quantum channel creation ($\downa{c}$), channel swaps ($\upa{c}$), classical message sender ($\csenda{\cmsg{c_1}}{\iota}$) and receiver ($\creva{\delta}{x}$), quantum message encoding ($\encode{c}{\mu}$) and decoding ($\decode{c}{x}$), local assembler ($\encodes{q}{\mu}{x}$) and extractor ($\decodes{q}{x}$).

$\encodes{q}{\mu}{x}$ and $\decodes{q}{x}$

Actions $\downa{c}$ and $\upa{c}$ manipulate membrane contexts; the former creates a shared channel, and the latter implements the quantum routing swap.
Actions $\csenda{\cmsg{c}}{\iota}$ and $\creva{\delta}{x}$ respectively send and receive messages and are inherited from $\Pi$-calculus \cite{DBLP:conf/concur/Sangiori93}.
The former sends a classical message through a classical channel $\cmsg{c}$, while the latter acts as a receiver waiting for a quantum or a classical message.
Actions $\encode{c}{\mu}$ and $\decode{c}{x}$ are \textit{global encoder and decoder} operations that manipulate a shared channel's content.
Actions $\encodes{q}{\mu}{x}$ and $\decodes{q}{x}$ are \textit{local assembler and extractor} that manipulate local quantum messages within a process.
Action $\msga{c}{x}$ transfers a shared channel between two membranes to become a channeled message, \textit{owned} by a process in a membrane.
QAM assumes that $\alpha$-conversions are applied automatically for variable bindings.
The variable $x$ in assemblers ($\encodes{q}{\mu}{x}$), decoders, and extractors ($\decode{c}{x}$ and $\decodes{c.\iota}{x}$), transfer operations $\msga{c}{x}$, and receipts ($\creva{\delta}{x}$) introduce new binding variables in the rest of the process.
}

\ignore{

The left membrane represents Alice and the right represents Bob.
Alice and Bob first collaboratively create a quantum channel $c$. Alice then encodes ($\encode{c}{e}$) a quantum message $e$ into the channel, which is decoded ($\decode{c}{x}$) by Alice immediately.
At this moment, Alice receives a classical residue $x$, while Bob receives a classical residue through a receiver operation ($\creva{c}{u}$). Alice then sends to Bob the classical residue through a classical sender $\csenda{\cmsg{c}}{x}$, and Bob locally assembles ($\encodes{z}{u}{v}$) it to recover the quantum message $e$.
The example hints about the QAM operations, which are discussed in detail below.

Alice has an additional process $\downa{c} 0$, representing the \textcolor{orange}{orange} party,
which collaboratively, with the first action of Alice's first process, creates a unique shared channel, named $c$. Alice then ($\encode{c}{\cmsg{e}}$) encodes a classical message $\cmsg{e}$ to the shared channel $c$.
Such an action is asynchronous, but its effect is global to membranes that share the channel.
After that, Alice asks for the ownership transformation ($\msga{c}{x}$) of the channel $c$,
which turns $c$ and its content ($\cmsg{e}$) as a message $x$.
Alice and Bob then create ($\downa{d}$) another shared channel $d$.
Immediately afterward,
Alice encodes the channeled message $c.\cmsg{e}$ into $d$, and decodes ($\decode{d}{y}$) the channel,
which is synchronized with the same channel receiver, e.g.,
$\creva{d}{u}$ in Bob's side;
once the synchronization happens, Alice ($y$) is assigned a classical
residue ($\cmsg{c.\cmsg{e}}$) and Bob ($u$) is assigned with a quantum residue
($\up{(c.\cmsg{e})}$).
Alice's final action ($\csenda{\cmsg{c}}{y}$) sends the classical residue to Bob through a classical channel $\cmsg{c}$;
once Bob receives ($\creva{\cmsg{c}}{z}$) the classical residue sent from Alice, he locally recovers the message $c.\cmsg{e}$ by assembling ($\encodes{z}{u}{v}$) the classical and quantum residues. Finally, he acquires message $\cmsg{e}$ by $\decodes{v}{w}$.
In QAM, a receiver, e.g., $\creva{d}{u}$, acts as a synchronous waiter waiting for a message through a channel $\delta$, either quantum or classical.
}

\section{Extending QAM to Remote Communication}\label{sec:case-study}

We extend the QAM to an evaluation framework, permitting remote communications for analyzing real-world HCQN protocols; their key assumptions are in \Cref{sec:hcqns}.
The extension's key includes location information in membranes, stabilizing channel's destinations, and computing probability success rates for a transmission path.

\subsection{Extending QAM Membranes with Locations}\label{sec:qamsyntax1}

\begin{figure}[t]
\vspace*{-2.5em}
{\centering
\small
  \[\begin{array}{l@{\;\;}lcl@{\quad}l@{\;\;}lcl} 
      \text{Locations} & g,h & \in & \mathbb{G} &
      \text{Intension ID} & n & \in & \mathbb{N} \\[0.2em]
      \text{Success Rate} & p & \in & [0,1] &
   \text{Relation Pair Multiset} & \theta
        & ::= & \overline{h \xleftrightarrow{p} g} \\[0.2em]

    \text{Intension Map} & \xi & ::= & n \to (g,h)
   &
           {\text{Predicates}} & {\varphi} & ::= & \xi * \theta * g * h \to \mathbb{T}
    \end{array}
  \]
}
\vspace*{-1em}
\caption{Extended QAM Syntax. Membranes are now $\scell{\overline{M}}_g$.}
  \label{fig:q-pi-semantics2}
\vspace*{-1.3em}
\end{figure}

As shown in 
\Cref{fig:q-pi-semantics2},
membrane structures are extended to contain location information, as $\scell{\overline{M}}_g$.
Every channel is established to transfer a message between two fixed locations.
We include the concept of intention ID $n\in \mathbb{N}$ for each channel and implement a channel $c$ as $c(n)$.
We also include an intention finite map $\xi$ mapping from intention IDs ($n$) to two locations $(g,h)$, indicating the target and the destination of the channel intended to establish.
In the paper, we abbreviate the channel $c(n)$ as $c$ and provide a function $\cn{im}(c)$ to access $c$'s intention ID, e.g., $\cn{im}(c)=n$.
Every execution is associated with a graph $\theta$, indicating the success rates of generating a quantum channel between two locations.
For example, in \Cref{fig:ctdan}, Alice and the router can create a channel with a $0.5$ success rate.
Function $\theta(g,h)$ produces the success rate connecting $g$ and $h$.

We now extend the QAM judgment to be $\overline{P}\xrightarrow{p(g,h)}_{\{\xi,\varphi,\theta\}} \overline{Q}$. Here, $\overline{P}$ and $\overline{Q}$ are configurations, with membranes extended with location information as in \Cref{fig:q-pi-semantics2}. The predicate $\varphi$ judges if a transition is valid by querying the transition labels.
In the new system, we modify the label to $p(g,h)$, meaning that a channel connects $g$ and $h$, and the success rate is  $p$.
The main utility of the extended QAM system is to evaluate the success rates of establishing long-distance channels.
Notice that the system is \textit{not} stochastic, e.g., Markov-chains \cite{markov1906,markov1907}.
Success rates $p$ in a label are viewed as a property of the transition edge instead of the edge's transition probability.

\vspace*{-1em}
{\small
\begin{mathpar}
\inferrule[Trans]{\overline{P} \xrightarrow{\kappa} \overline{Q} \\\neg \cn{is\_q}(\cn{chan}(\kappa))}
       {\overline{P} \xrightarrow{\kappa}_{\{\xi,\varphi,\theta\}} \overline{Q}  }

\inferrule[Chan]{\overline{P} \xrightarrow{c} \overline{Q}\\ \cn{loc}(\overline{P},c)=(g,h)\\\varphi(\xi, \theta,g,h)}
       {\overline{P} \xrightarrow{(\theta(g,h))(g,h)}_{\{\xi,\varphi,\theta\}} \overline{Q}  }
  \end{mathpar}
\begin{mathpar}
\inferrule[Msg]{\overline{P} \xrightarrow{c.\mu} \overline{Q}\\\cn{loc}(\overline{P},c,\cn{chan}(\mu))=(g,h)\\\varphi(\xi, \theta,g,h)}
       {\overline{P} \xrightarrow{(1)(g,h)}_{\{\xi,\varphi,\theta\}} \overline{Q}  }
  \end{mathpar}
}
\vspace*{-0.6em}

Similar to the extendability in CHAM, we extend QAM to include the three rules. Rule \rulelab{Trans} turns a QAM transition, without quantum channel labels (checked by $\cn{is\_q}$), into a transition in the extended QAM.
When creating a channel, Rule \rulelab{Chan} finds the rate in the intransitive set $\theta$ for the edge $(g,h)$,
where $g$ and $h$ are found by the function $\cn{loc}(\overline{P},c)$ searching the two locations that $c$ resides.
Rule \rulelab{Msg} deals with the case when a decoding happens. In QES (\Cref{sec:enswaps}), the ownership of a quantum channel $\cn{chan}(\mu)$ is transferred to a third party,
so we need to access the location information of the two new parties through the function $\cn{loc}(\overline{P},c,\cn{chan}(\mu))$, producing the two different locations that the two channels $c$ and $\cn{chan}(\mu)$ are connected to. In this process, the success rate of the transition is $1$.

Consider \Cref{def:example6}, we extend the membranes from $\scell{-}$ to $\scell{-}_g$, where $g$ can be $A$, $R$, and $B$, referring to the locations of Alice, the router, and Bob.
We include intention IDs $n$ and $n'$ for channels $c$ and $d$. In $\xi$, $n$ maps to $(A,B)$ and $n'$ maps to $(R,B)$.
We also need to implement the $\theta$ map as the diagram shown in \Cref{fig:ctdan} as:

\vspace*{-0.6em}
{\small
\begin{center}
$
\theta \triangleq 
\{A\xleftrightarrow{0.5} R, R \xleftrightarrow{0.5} B\}
\qquad
\xi \triangleq 
\{n\to (A,B), n' \to (R,B)\}
$
\end{center}
}
\vspace*{-0.6em}

Here, we can set the predicate $\varphi(\xi, \theta,g,h)$ always to answer \cn{true}, provided that $(g,h)$ is a valid edge in $\theta$.
At the time when the channels $c$ and $d$ are being created (\Cref{def:example6}), we apply rule \rulelab{Chan} twice, and they raise labels $0.5(A,R)$ and $0.5(R,B)$.
At the time when we decode the channel $d$ (\Cref{def:example7}), with the transition label being $d.c.\circ$,
we apply rule \rulelab{Msg} to prolong the channel $d$ to connect the path between $A$ and $B$, with the success rate $1$.
Along the chain of the successful creation of a long-distance channel from $A$ to $B$, if we compute the total success rate along the above three steps, it is $0.5* 0.5 * 1=0.25$.

\subsection{Defining the QPass and QCast Protocols}\label{sec:case-qpass}

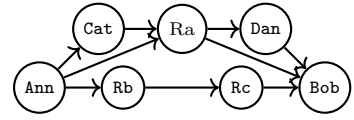
\begin{wrapfigure}{r}{4.8cm}
{\scriptsize
    \begin{tikzpicture}[node distance={1.1cm}, thick, main/.style = {draw, circle}] 
    \node[main] (1) {\cn{Cat}}; 
    \node[main] (2) [right of=1] {Ra};
    \node[main] (3) [right of=2] {\cn{Dan}};
    \node[main] (4) [below left of=1] {\cn{Ann}};
    \node[main] (5) [below right of=3] {\cn{Bob}};
    \node[main] (6) [below left of=2] {\cn{Rb}};
    \node[main] (7) [below right of=2] {\cn{Rc}};
    \draw[->] (1) -- node[midway, above] {} (2); 
    \draw[->] (2) -- node[midway, above] {} (3);  
    \draw[->] (4) -- node[midway, above left] {} (1);  
    \draw[->] (4) -- node[midway, above] {} (6);  
    \draw[->] (6) -- node[midway, above] {} (7);  
    \draw[->] (7) -- node[midway, above] {} (5); 
    \draw[->] (3) -- node[midway, above right] {} (5);   
    \draw[->] (4) -- node[near end, below left] {} (2);   
    \draw[->] (2) -- node[near start, below right] {} (5);   
    \end{tikzpicture}
    \caption{Example path graph.}
    \label{fig:anndan}
}
\end{wrapfigure}

The QPass and QCast protocols \cite{10.1145/3387514.3405853} define the routing behaviors of HCQN communications.
In defining the two protocols, we utilize the above-extended QAM semantics by modifying the path graph $\theta$ for a specific routing table and defining the predicate $\varphi$ for a specific protocol.
Here, we model the simplified versions of the QPass and QCast protocols,
where the simplified QPass protocol permits a creation (applying a rule \rulelab{Chan}) and extension (applying a rule \rulelab{Msg}) of a quantum channel $c$, only if the edge that the channel connects, as $(g,h)$, are in the shortest path between the starting and the final destinations ($\xi(\cn{im}(c))$).
The simplified QCast protocol permits the creation or extension of a channel if the edge $(g,h)$ is in the path between the starting and the final destinations, where the sum of the assumed success rates is maximized.

\vspace*{-0.5em}
{\small
\begin{center}
$
\theta(\overline{p})\triangleq\{\cn{Ann}\xleftrightarrow{p_1} \cn{Cat}, \cn{Ann}\xleftrightarrow{p_2} \cn{Ra}, \cn{Ann}\xleftrightarrow{p_3} \cn{Rb}, \cn{Cat}\xleftrightarrow{p_4} \cn{Ra}, \cn{Cat}\xleftrightarrow{p_5} \cn{Ra},  \cn{Ra}\xleftrightarrow{p_6} \cn{Dan},  \cn{Ra}\xleftrightarrow{p_7} \cn{Bob}, \cn{Dan}\xleftrightarrow{p_8} \cn{Bob}, \cn{Rb}\xleftrightarrow{p_9} \cn{Rc}, \cn{Rc}\xleftrightarrow{p_{10}} \cn{Bob}\}
$
\end{center}
}
\vspace*{-0.5em}

As an example, we define an example routing map (\Cref{fig:anndan}) as $\theta(\overline{p})$ above, with $\overline{p}$ being a list of assumed success rates for creating the channel between two parties in an edge. 
In the example map, for simplicity, we define the destination to be $\cn{Bob}$ for every quantum channel.
We discuss the predicate definition for the QPass and QCast protocols below.

\noindent\textbf{The QPass Protocol.}
In defining the QPass protocol, the predicate $\varphi_p$ is defined as:

\vspace*{-0.5em}
{
\begin{center}
$\varphi_p(\xi, \theta(\overline{p}), g, h)\triangleq (g,h) \in \cn{shortest}(\theta(\overline{p}),\xi(\cn{im}(c)))$
\end{center}
}
\vspace*{-0.5em}

A valid quantum channel creation or extension is validated if the edge generated from a transition is in the shortest path from the left node of the channel to the destination.
An example of transmitting a message from \cn{Ann} to \cn{Bob} in \Cref{fig:anndan} is shown below.

\vspace*{-0.5em}
\begin{example}[QPass Example Code]\label{def:example8}\rm
Below, we reuse the code in \Cref{def:example6} to define the behaviors of the membranes \cn{Ann}, \cn{Ra}, and \cn{Bob} for transmitting a message $\cmsg{c_1}.e$ from \cn{Ann} to \cn{Bob} via the intermediate router \cn{Ra}.
Here, $A={{\seq{ \decode{c}{x}}{\seq{{\csenda{a}{x}}}{0}}}}$, $R=\decode{d}{y}\seq{\csenda{a_1}{y} }{0}$, and $B=\seq{\creva{c}{u}\creva{a}{t}\encode{u}{t}}{0}$.

{\small
{
$
\pard{\pard{\bscell{\pard{\seq{\downa{c}}{\seq{ \encode{c}{\cmsg{c_1}}}{A}}}{\pard{\cmsg{c_1}.e}{\circ}}}_{\cn{Ann}} }
{{\bscell{\pard{\seq{\downa{c}\downa{d}\encode{d}{c}}{R}}{\pard{\circ}{\circ}}}_{\cn{Ra}} }}} 
{
{{\bscell{\pard{\seq{\downa{d}}{\seq{ \creva{d}{w}}{\seq{\creva{a_1}{z}}{\seq{\encode{w}{z}}{B }}}}}{\circ}}_{\cn{Bob}} }}}
$
}
}
\end{example}
\vspace*{-0.5em}

The above example defines a small portion of the system in \Cref{fig:anndan}, where we send a message from $\cn{Ann}$ to $\cn{Bob}$ via the router $\cn{Ra}$, as we use the same $\xi$ map above.
The transitions of \Cref{def:example8} are listed as follows.

\vspace*{-0.5em}
\begin{example}[QPass Example Transitions]\label{def:example9}\rm

{\noindent\small
\[
\begin{array}{ll}
...\xrightarrow{p_2(\cn{Ann},\cn{Ra})}
&
\pard{\pard{\bscell{\pard{A}{c.\cmsg{c_1}.e}}_{\cn{Ann}}}
{\bscell{\pard{\pard{\seq{\downa{d}\encode{d}{c}}{R}}{\circ}}{{c.\circ}}}_{\cn{Ra}}}}
{{\bscell{\pard{\seq{\downa{d}\creva{d}{w}}{\seq{\creva{a_1}{z}}{\seq{\encode{w}{z}}B }}}{\circ}}_{\cn{Bob}}}}
\\[0.3em]

...\xrightarrow{p_7(\cn{Ra},\cn{Bob})}
&
\pard{\pard{\bscell{\pard{A}{c.\cmsg{c_1}.e}}_{\cn{Ann}}}
{\bscell{\pard{\pard{\seq{\encode{d}{c}}{R}}{d.\circ}}{{c.\circ}}}_{\cn{Ra}}}}
{{\bscell{\pard{\seq{ \creva{d}{w}}{\seq{\creva{a_1}{z}}{\seq{\encode{w}{z}}B }}}{d.\circ}}_{\cn{Bob}}}}
\\[0.3em]
...\xrightarrow{1(\cn{Ann},\cn{Bob})}
&
\pard
{\bscell{\pard{A}{c.\cmsg{c_1}.e}}_{\cn{Ann}}}
{
\pard{{\bscell{{{\seq{\csenda{a_1}{\cmsg{d}.\cmsg{c.\circ}} }{0}}}}_{\cn{Ann}}}}
{{\bscell{\pard{\seq{\creva{a_1}{z}}{\seq{\encode{\cmsg{d}}{z}}{B} }}{\cmsg{d}.\up{c.\circ}}}_{\cn{Ann}}}}}
\\[0.2em]
...
\end{array}
\]
}
\end{example}
\vspace*{-0.5em}

In the above example, the creation of the channels $c$ and $d$ are validated in QPass,
because the locations where the channels are connecting belong to the shortest path from \cn{Ann} to \cn{Bob} and $\cn{Ra}$ to $\cn{Bob}$ in \Cref{fig:anndan}.
The last step performs a decoding, as in \Cref{def:example7},
and it is essentially a step to convey a party of the channel $c$ from $\cn{Ra}$ to $\cn{Bob}$, stated in \Cref{def:decoding-good}.
This procedure is validated in QPass by applying rule \rulelab{Msg} with the predicate definition $\varphi_p$ above,
because the edge $(\cn{Ann},\cn{Bob})$ is also in the shortest path from \cn{Ann} to \cn{Bob}.

\noindent\textbf{The QCast Protocol.}
In defining the QCast protocol, the predicate $\varphi_c$ is defined as:

\vspace*{-0.5em}
{\small
\begin{center}
$\varphi_c(\xi, \theta(\overline{p}), g, h)\triangleq (g,h) \in \cn{max}(\theta(\overline{p}),\xi(\cn{im}(c)))$
\end{center}
}
\vspace*{-0.5em}

In the case of QCast, we validate a quantum channel creation or relocation,
if the edge generated from a transition is in the path that maximizes the possibility of sending a message from the source to the target.
We examine the QCast behavior by reusing the example in \Cref{def:example8} as the following.

\vspace*{-0.5em}
\begin{example}[QCast Example Transitions]\label{def:example10}\rm

{\noindent\small
\[
\begin{array}{lll}
...\xrightarrow{p_2(\cn{Ann},\cn{Ra})}
&
\pard{\pard{\bscell{\pard{A}{c.\cmsg{c_1}.e}}_{\cn{Ann}}}
{\bscell{\pard{\pard{\seq{\downa{d}\encode{d}{c}}{R}}{\circ}}{{c.\circ}}}_{\cn{Ra}}}}
{{\bscell{\pard{\seq{\downa{d}\creva{d}{w}}{\seq{\creva{a_1}{z}}{\seq{\encode{w}{z}}B }}}{\circ}}_{\cn{Bob}}}}
&
\textcolor{red}{\textbf{?}}
\\[0.5em]

...\xrightarrow{p_7(\cn{Ra},\cn{Bob})}
&
\pard{\pard{\bscell{\pard{A}{c.\cmsg{c_1}.e}}_{\cn{Ann}}}
{\bscell{\pard{\pard{\seq{\encode{d}{c}}{R}}{d.\circ}}{{c.\circ}}}_{\cn{Ra}}}}
{{\bscell{\pard{\seq{ \creva{d}{w}}{\seq{\creva{a_1}{z}}{\seq{\encode{w}{z}}B }}}{d.\circ}}_{\cn{Bob}}}}
&
\textcolor{red}{\textbf{?}}
\\[0.5em]
...\xrightarrow{1(\cn{Ann},\cn{Bob})}
&
\pard
{\bscell{\pard{A}{c.\cmsg{c_1}.e}}_{\cn{Ann}}}
{
\pard{{\bscell{{{\seq{\csenda{a_1}{\cmsg{d}.\cmsg{c.\circ}} }{0}}}}_{\cn{Ann}}}}
{{\bscell{\pard{\seq{\creva{a_1}{z}}{\seq{\encode{\cmsg{d}}{z}}{B} }}{\cmsg{d}.\up{c.\circ}}}_{\cn{Ann}}}}}
&
\textcolor{red}{\textbf{?}}
\\[0.5em]
...
\end{array}
\]
}
\end{example}
\vspace*{-0.5em}

In the above example, the three transitions are marked with question marks because their validity depends on the success rate values ($\overline{p}$) given in $\theta(\overline{p})$.
For example, if we set $p_2$, $p_9$ and $p_{10}$ to be $0.5$, and the other success rates to be $0.3$, the maximized success rate path is the path: $\cn{Ann} \to \cn{Rb}\to \cn{Rc} \to \cn{Bob}$.
Thus, the transitions in \Cref{def:example10} are all invalidated because the edges are not in the maximized success rate path.
One thing worth noting is that the success rates ($\overline{p}$) are fixed before the program execution,
and they represent the guesses of success rates to construct channels between two distinct locations.
The original QCast paper has an algorithm to adjust the success rates when network transmissions are evolving, which will be a future study of the QAM.

\section{Behavioral Refinement and Equivalence} \label{sec:refinement}

To illustrate the QAM, we use the traditional trace refinement,
Here, we define trace-refinement relations on labeled transition systems (LTS) for the QAM and the extended QAM, following traditional notions in process algebra.
We view the QAM system as LTS $(\Cs,\As,\rightarrow)$, with $\Cs$ being the set of configurations and $\As$ containing transition labels ($\kappa$),
while the Extended QAM is another LTS $(\Cs_d,\As_d,\rightarrow_{\{\xi,\varphi,\theta\}})$,
with $\Cs_d$ being the set of configurations extended with location information, and $\As_d$ containing transition labels, having the form $p(g,h)$ in the extended QAM.
We also use $::$ as a label sequence concatenation operation and $\epsilon$ as the empty sequence.
We first define QAM finite traces and trace refinement based on Back and von Wright's definition \cite{10.1007/978-3-540-48654-1_28}, which provides an abstraction of transitions.

\vspace*{-0.5em}
\begin{definition}[QAM Finite Trace Set]\label{def:qamtraces}\rm
Given the QAM LTS $(\Cs,\As,\rightarrow)$, the trace set $\Tts(C)$ for configuration $C\in\Cs$ is defined inductively as:
  \begin{itemize}
  \item $\epsilon \in \Tts(C)$.
  \item $t \in \Tts(C)$, given $C \longrightarrow C'$ and $t \in \Tts(C')$.
  \item $\kappa::t \in \Tts(C)$, given $C \xrightarrow{\kappa} C'$ and $t \in \Tts(C')$.
  \end{itemize}

\end{definition}
\vspace*{-0.5em}

\begin{definition}[QAM Trace Refinement]\label{def:traceeq}\rm
Given $\Qs=(\Cs,\As,\rightarrow)$, and $C\;C' \in \Cs$, we say that $C$ trace refines $C'$, written as $C \sqsubseteq_{\Qs} C'$, iff $\Tts(C)\subseteq \Tts(C')$.
\end{definition}
\vspace*{-0.5em}

The above trace-refinement definition is suitable for QAM programs.
For an extended QAM execution, success rates can be attached to a transition label.
We need to extend the above trace-refinement relation to include probability success rates, treated as a special property of the labels.
We first define a likeliness finite trace set.

\begin{definition}[Likeliness Finite Trace Set]\label{def:dqamtraces}\rm
Given $\xi$, $\theta$, and an extended QAM system $\Qs=(\Cs_d,\As_d,\rightarrow_{\{\xi,\varphi,\theta\}})$, the trace set $\Tts(\Qs,C)$ over $C\in\Cs_d$ is defined inductively as:

\vspace*{-0.5em}
  \begin{itemize}
  \item $\epsilon \in \Tts(\Qs,C)$.
  \item $t \in  \Tts(\Qs,C)$, given $C \darrow{\xi,\varphi,\theta} C'$ and $t \in  \Tts(\Qs,C')$.
  \item $p(g,h)::t \in \Tts(C)$, given $C \xrightarrow{(p)(g,h)}_{\{\xi,\varphi,\theta\}} C'$ and $t \in \Tts(C')$.
  \end{itemize}

\end{definition}
\vspace*{-0.5em}

The likeliness trace ordering relation and QAM trace refinement can now be defined as:

\vspace*{-0.5em}
\begin{definition}[Likeliness Trace Order]\label{def:ptraceeq}\rm
Two QAM finite traces, $\sigma_1$ and $\sigma_2$, are likely ordered, written as $\sigma_1 \preceq \sigma_2$, iff the following holds:

\vspace*{-0.5em}
  \begin{itemize}
  \item $\sigma_1 = \epsilon$ and $\sigma_2 =\epsilon$.
  \item If $\sigma_1 = p(g,h) :: \sigma_1'$, then $\sigma_2 =p'(g,h) :: \sigma_2'$ for some $p'$, and $p \le p'$ and $\sigma_1' \preceq \sigma_2'$.
  \end{itemize}
\end{definition}
\vspace*{-0.5em}

\begin{definition}[Likeliness Trace Subset]\label{def:ptracesub}\rm
Given two likeliness trace sets $S$ and $S'$, $S \subseteq_r S'$ iff, for every $\sigma \in S$, there is $\sigma' \in S'$, such that $\sigma \preceq \sigma'$.
\end{definition}
\vspace*{-0.5em}

\begin{definition}[Likeliness QAM Trace Refinement (LQTR)]\label{def:ptracerefine}\rm
Given $\xi$, $\theta$, and two extended QAM systems $\Qs_1=(\Cs_d,\As_d,\rightarrow^1_{\{\xi,\varphi_1,\theta\}})$ and $\Qs_2=(\Cs_d,\As_d,\rightarrow^2_{\{\xi,\varphi_2,\theta\}})$, and configurations $C_1 \; C_2 \in \Cs$,  we say that $(\Qs_1,C_1)$ likely trace refines $(\Qs_2,C_2)$, written as $(\Qs_1,C_1) \sqsubseteq_r (\Qs_2,C_2)$, iff $\Tts(\Qs_1,C_1)\subseteq_r \Tts(\Cs_2,C_2)$.

\end{definition}
\vspace*{-0.5em}

The LQTR relation is defined based the QAM trace refinement relation, with the modification of labels,
i.e., the extended QAM transition definition in \Cref{sec:qamsyntax1} extends the transition rules in the QAM.
Notice above that the configuration and label sets $\Cs_d$ and $\As_d$ as well as maps $\theta$ and $\xi$ are the same in the two systems.
\ignore{
\vspace*{-0.5em}
\begin{theorem}[LQTR and QTR Trace Containment]\label{def:traceeq1}\rm
Given $C_1$ and $C_2$, $\Qs_d=(\Cs_d,\As_d,\rightarrow_{\{\xi,\varphi,\theta\}})$, such that $(\Qs_d,C_1) \sqsubseteq_r (\Qs_d,C_2)$, 
we then have $C'_1 \sqsubseteq_{\Qs} C'_2$ with $\Qs=(\Cs,\As,\rightarrow)$, and $C_1$ and $C_2$ are equivalent to $C'_1$ and $C'_2$ but removing membrane location information. 
\end{theorem}
\vspace*{-0.5em}

We show above that LQTR is a trace containment of QTR, i.e., every trace in LQTR is a prefix of a trace in QTR.
To relate the configurations in the two systems, we simply remove the location information associated with membranes.
}

The LQTR definition is useful in relating two different HCQN protocols. For example, we can relate the above QPass and QCast protocols to the theorem below:

\vspace*{-0.5em}
\begin{theorem}[QPass and QCast Trace Refinement]\label{def:traceeq2}\rm
Given $\xi$, $\theta(\overline{p})$, and an initial configuration $C\in \Cs_d$ and the QPass and QCast systems as stated in \Cref{sec:case-qpass},
there is an assignment $\overline{p}$, such that (QPass,$C$) $\sqsubseteq_r$ (QCast,$C$).

\end{theorem}
\vspace*{-0.5em}

\section{Related Work}
\label{Related Work}

\noindent\textbf{Quantum Process Algebra.}
The QAM's design was inspired by several existing quantum process calculi:
qCCS \cite{10.1145/1507244.1507249,10.1145/2400676.2400680,10.1007/978-3-030-45237-7_2},
Communicating Quantum Processes (CQP) \cite{10.1145/1040305.1040318},
quantum model checker (MQC) \cite{mqcwork,mqcworkexample},
QPAlg \cite{10.1145/977091.977108}, and eQPAlg \cite{haider2020extended}.
These process calculi enhance existing message-passing models, such as CSP \cite{Hoare:1985:CSP:3921} and $\Pi$-calculus \cite{MILNER19921}, to define HCQN protocols.
They do not intend to provide abstract semantics and a new communication model for specifically describing HCQNs.
First, they did not define the unique properties of quantum resources.
For example, qCCS allows the message passing of variables referring to quantum resources and permits their remote controls; thus, HCQN protocols that do not meet the physical realities can be defined in qCCS. Second, their message-passing views did not capture the lifetime properties of quantum channels, as we describe in \Cref{sec:design-resource}.
Third, all of the previous quantum process calculi are not based on abstract semantics, which is the mathematical characterization of program behaviors.
They mainly admit concrete operational semantics instead of abstract ones by integrating quantum circuits and states in the syntax of some process algebra, e.g., CQP and qCCS extend CSP \cite{Hoare:1985:CSP:3921} and $\Pi$-calculus \cite{MILNER19921} with operations describing quantum circuits and semantic states that do not appear in their language operation syntax. 
The key to process calculi is to provide insight into HCQN communications.
Although quantum circuit semantics provides a thorough story of HCQN protocols, many subtleties in these protocols might be glossed over, e.g., the projective nondeterminism.
It would be unfortunate for protocol designers to examine the projective nondeterminism behavior through a heavy study of quantum circuit states,
especially considering the fact that projective nondeterminism is essential in HCQN communications, and a system should show the behavior directly in front of users.

\noindent\textbf{Traditional Process Algebra.}
Communicating Sequential Processes (CSP) \cite{Hoare:1985:CSP:3921} and $\Pi$-calculus \cite{MILNER19921} are process calculi suitable for defining concurrent systems based on the message-passing model.
Several bisimulation and trace-refinement protocol verification methodologies exist for CSP and the $\Pi$-calculus \cite{FDR2,fdr3,DBLP:conf/concur/Sangiori93}.
As noted earlier, the Chemical Abstract Machine \cite{BERRY1992217} and Linda \cite{1663305} are the inspiration of the QAM.
Membrane computing \cite{PAUN2000108}, Modulo structural operational semantics \cite{MOSSES2004195}, and the K framework \cite{rosu-serbanuta-2010-jlap}
are three general language definition frameworks for defining language semantics formally.

\noindent\textbf{Quantum Network Protocols.}
Quantum teleportation \cite{PhysRevLett.70.1895,Rigolin_2005} serves as the basis for quantum communication between two parties. Julia-Diaz {\it et al.} \cite{twoqubittele} provides a two-qubit quantum teleportation protocol. Superdense coding \cite{PhysRevLett.69.2881} encodes a classical message into a quantum channel.
Quantum routing investigates long-distance quantum message transmission, with QES being one of the promising protocols for the task \cite{fundamentallimits,aam9288,10.1145/3386367.3431293}.
QPass and QCast are protocols based on the quantum-swap algorithm \cite{10.1145/3387514.3405853} to maximize the transmission chances, through static and semi-dynamic analyses.
Researchers developed their circuit implementations \cite{PhysRevResearch.4.043064,10.1145/3341302.3342070} and new protocols for enhancing the reliability \cite{Pirker_2019}.
Chakraborty {\it et al.} \cite{https://doi.org/10.48550/arxiv.1907.11630} provided an alternative routing protocol to permit distributed routing.
Li {\it et al.} \cite{10.48550/arxiv.2111.07764} and Caleffi \cite{8068178} provide systems to improve transmission chances and message delivery rates.
The refinement of HCQN network protocols strongly motivates the QAM's development.

\section{Conclusion and Future Work} \label{sec:conclusions}

We propose the Quantum Abstract Machine (QAM), admitting an abstract semantics, as a communication model for HCQN communications,
which captures the important aspects of HCQNs and permits the definition and analysis of HCQN protocols;
they are different from the behaviors of classical network communications based on message-passing models.
Many HCQN properties, such as \Cref{def:no-cloning-good,def:decoding-good}, are surfaced in the QAM due to our abstract semantics. 
We also show the QAM semantics can be extended to an evaluation framework suitable for quantum resource analyses for real-world HCQN protocols, such as QPass and QCast,
and utilize classical trace refinement relations, as an example of the QAM's extendability for analytical tools, to analyze real-world protocols.

As one of the future work, we plan to establish a temporal logic model-checking environment on top of the QAM to verify HCQN protocols.
Supporting HCQN security protocol analyses will also be a future direction for the QAM.
The encoding and meet operations ($\odot$) in the QAM are extensible for the task of defining different quantum operations for enforcing cryptographic methodologies.
One example extension is given in \Cref{sec:qamsecurity}.
In addition, we plan to develop a concurrent quantum circuit language that inherits the properties and restrictions from the QAM so that every defined physical level protocol in the language is an implementable one following physical laws. A small example is given in \Cref{sec:qamsemproof}. 
 
\bibliography{reference}

\newpage
\appendix

\section{Background:  Quantum Computing and HCQNs}
\label{sec:background}

\noindent\textbf{\textit{Quantum States.}} A quantum state consists of one or more quantum bits (\emph{qubits}). A qubit can be expressed as a two-dimensional vector $\begin{psmallmatrix} \alpha \\ \beta \end{psmallmatrix}$; $\alpha,\beta$ are complex numbers such that $|\alpha|^2 + |\beta|^2 = 1$.  The $\alpha$ and $\beta$ are called \emph{amplitudes}. 
We frequently write the qubit vector as $\alpha\ket{0} + \beta\ket{1}$, with $\ket{0} = \begin{psmallmatrix} 1 \\ 0 \end{psmallmatrix}$ and $\ket{1} = \begin{psmallmatrix} 0 \\ 1 \end{psmallmatrix}$. When both $\alpha$ and $\beta$ are non-zero, we can think of the qubit as being ``both 0 and 1 at once,'' a.k.a. a \emph{superposition}. For example, $\frac{1}{\sqrt{2}}(\ket{0} + \ket{1})$ is an equal superposition of $\ket{0}$ and $\ket{1}$. 
We can join multiple qubits together to form a larger quantum state using the \emph{tensor product} ($\otimes$) from linear algebra. For example, the two-qubit state $\ket{0} \otimes \ket{1}$ (also written as $\ket{01}$) corresponds to vector $[~0~1~0~0~]^T$. 
If a multi-qubit state cannot be expressed as the tensor of individual states; such states are called \emph{entangled}. One example is the \emph{Bell pair} $\frac{1}{\sqrt{2}}(\ket{00} + \ket{11})$, which is frequently used as a quantum channel in HCQN protocols as the two qubit elements in the pair are used by two parties to synchronize information. 
Entangled states lead to exponential blowup: a general $n$-qubit state must be described with a $2^n$-length vector rather than $n$ vectors of length two.

\noindent\textbf{\textit{Quantum Computation: From Circuits to Task Blocks.}}
Quantum computation is essentially hybrid, containing both quantum and classical components so that they can collaboratively finish a task (the \emph{QRAM model}~\cite{Knill1996}). Computation on a large quantum state consists of a series of \emph{quantum operations}, each of which acts on a subset of qubits in the state. 
In the standard presentation, quantum computations are expressed as \emph{circuits}, e.g., \Cref{fig:channel-circuit-example}, in which each horizontal wire represents a qubit, boxes on the wires indicate quantum operations, or \emph{gates}, and \textbf{bold} lines represent classical hardware.
Applying a gate to a state \emph{evolves} the state.
The traditional semantics is expressed by multiplying the state vector by the gate's corresponding matrix representation; $n$-qubit gate are $2^n$-by-$2^n$ matrices. 
Except for measurement gates, a gate's matrix must be \emph{unitary}, preserving the unitarity invariant of quantum states' amplitudes. 
A \emph{measurement} operation extracts classical information from a quantum state. Measurement collapses the quantum state to a basis state with a probability related to the state's amplitudes. For example, measuring $\frac{1}{\sqrt{2}}(\ket{0} + \ket{1})$ collapses the state to $\ket{0}$ with probability $\frac{1}{2}$ and likewise for $\ket{1}$, returning classical value $0$ or $1$, respectively. 
In a Bell pair $\frac{1}{\sqrt{2}}(\ket{00} + \ket{11})$, measuring one qubit results in having the same bit appearing in the other, e.g., if the first bit is measured as $0$, the second-bit results in $0$; such property enables Bell pairs being quantum channels.

While circuit expressions are common, researchers utilized \emph{task blocks}, consisting of groups of gates, to describe algorithms.
Such blocks appear repetitively in quantum circuits, such as the \cn{Bell} block appearing in \Cref{fig:channel-circuit-example,fig:channel-block-example} for preparing a Bell pair by applying the block on a two-qubit state $\ket{0}\otimes\ket{0}$,
 and the \cn{Decode} block decodes a quantum channel to send a quantum message from one party to the other.
We model these repetitive task blocks, appearing in many HCQN protocols, as operations in the QAM, regardless of their circuit details, such as qubit size.

\ignore{ 
\begin{figure}[t]
{\centering
\hspace*{-1.2em}

\begin{minipage}[b]{.5\textwidth}
                 \includegraphics[width=1\textwidth]{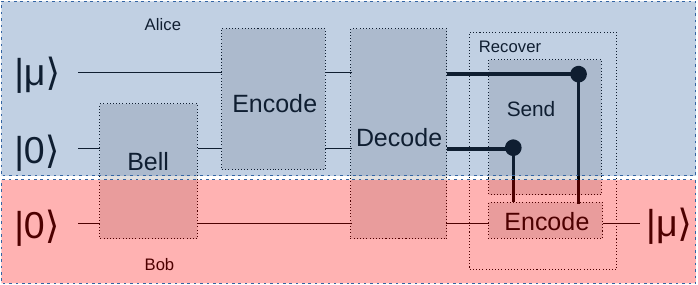}
            \caption{Teleportation Task Blocks}
            \label{fig:background-circuit-examplea}
 \end{minipage}
\hfill{}
\begin{minipage}[b]{.45\textwidth}
                 \includegraphics[width=1\textwidth]{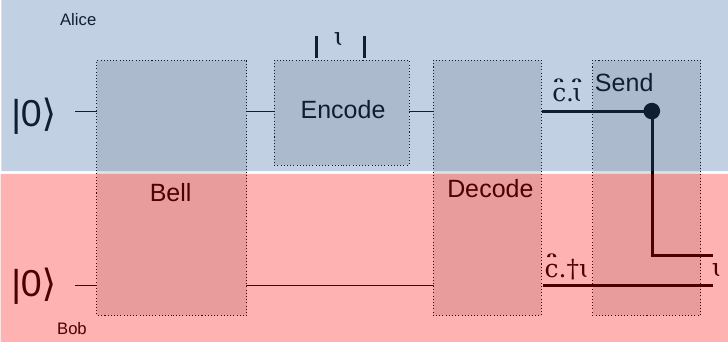}
            \caption{Superdense Coding Task Blocks}
            \label{fig:background-circuit-exampleb}
 \end{minipage}
}
\end{figure}
}
\noindent\textbf{\textit{No Cloning Theorem.}} 
The \emph{no-cloning theorem} suggests no general way of copying a
quantum value. In quantum circuits, this is related to ensuring the reversible property of unitary gate applications.
For example, the control-not gate in the Bell pair circuit (\Cref{fig:channel-circuit-example}) cannot have two ends to refer to the same qubit.
There is also no general way of copying a quantum qubit if its state is arbitrarily defined. In the classical substitution, one can substitute a value for all occurrence of a variable in a term,
which results in creating possible multiple copies of quantum qubits, violating the no-cloning law.
However, if the quantum qubit state is known, there can be ways to create its copies.
For example, our replications in \Cref{sec:qamsem} reflect the procedure of locally copying a known quantum state.
Quantum channel creation is also linked to no-cloning, as the discussion below.

\begin{figure}[t]
{\centering
  \includegraphics[width=0.7\textwidth]{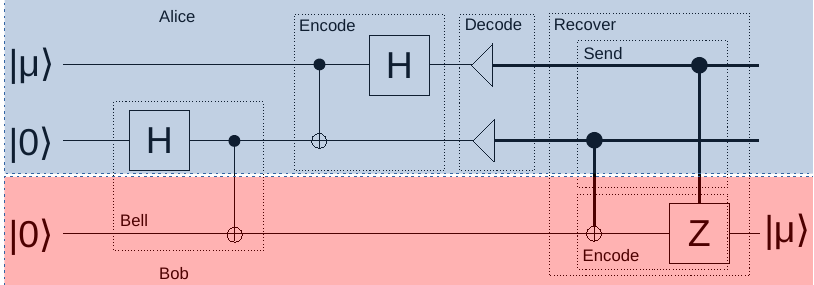}
            \caption{Teleportation Circuit}
            \label{fig:background-circuit-exampleac}
}
\end{figure}

\noindent\textbf{\textit{HCQNs.}}
\label{sec:hcqns}
The essence of HCQNs is to hybridize the existing classical network infrastructure to construct the next generation of communication networks, a.k.a. quantum internet, featured with quantum mechanics \cite{Granelli2022}. HCQN can provide more secure message communications than the existing infrastructure due to the \emph{no-cloning} theorem, i.e., quantum messages cannot be cloned, so hackers have no way to eavesdrop without being realized by users.

Essentially, HCQN communications view quantum entanglement behaviors as quantum channels. Encoding messages in a party of a quantum channel leads to a change in the quantum state of the whole channel due to quantum entanglement. If the two parties of a quantum channel are located in different places, then the change in one party affects the other.
This phenomenon is utilized to construct the quantum teleportation protocol \cite{PhysRevLett.70.1895,Rigolin_2005}, whose circuit representation is in \Cref{fig:background-circuit-exampleac}.
In the protocol, a message is encoded in a party of a quantum channel; such behavior affects the whole channel.
Then, the protocol decodes the other party to access the message. In this process, the information and quantum resources in the first party disappear, so the no-cloning theorem is preserved.
In quantum teleportation, the message being encoded is a quantum one, and people utilize protocols such as superdense coding \cite{PhysRevLett.69.2881}, to generate a quantum message that encodes classical information.

HCQN communications have certain physical limitations. First, quantum qubits cannot be relocated, i.e., one cannot send out a qubit directly. Instead, it is the information in the qubit being sent out.
In addition, there is a few patterns for permitting the remote controls of quantum qubits, i.e., there is a serious restriction to apply a quantum operation to a qubit in a distinct location.
For example, the rate of building a quantum channel connecting three remote locations is very low and has only lab experiment significance \cite{de_Jong_2023}.
In addition, the no-cloning restriction can be violated very easily in the setting of many classical operations, such as the classical substitution in \Cref{sec:design-resource}.
Finally, it is worth noting that projective nondeterminism is caused by the mapping from the quantum world to the classical world.
When a quantum channel $c$ is decoded, it generates a pair of classical and quantum residues.
Assume that we decode two quantum channels $c_1$ and $c_2$  with the same content $q=\cmsg{d}.e$, which results in two pairs $\cmsg{c_1}.\cmsg{q}$ and $\cmsg{c_1}.\up{q}$ as well as $\cmsg{c_2}.\cmsg{q}$ and $\cmsg{c_2}.\up{q}$. One cannot recover $q$ through the composition of wrong pairs, i.e., both $(\cmsg{c_1}.\cmsg{q} \odot \up{c_2}.\up{q})$ and $(\cmsg{c_2}.\cmsg{q} \odot \up{c_1}.\up{q})$
are not likely to be equal to $q$ i.e., if $q$ is $n$ qubit, then the probability for $(\cmsg{c_1}.\cmsg{q} \odot \up{c_2}.\up{q})$ to be equal to $q$ is $\frac{1}{2^n}$.

The above summarizes the reasons why many HCQN protocols is originated from the quantum teleportation and superdense coding protocols and have similar structures as the two.
The two algorithms describe network communication within a limited distance.
To permit long-distance communication, the quantum entanglement swap protocol plays a key role in prolonging a quantum channel to permit reaching distinct locations, which is elaborated in \Cref{sec:enswaps}.

\begin{wrapfigure}{r}{4.3cm}
\includegraphics[width=0.35\textwidth]{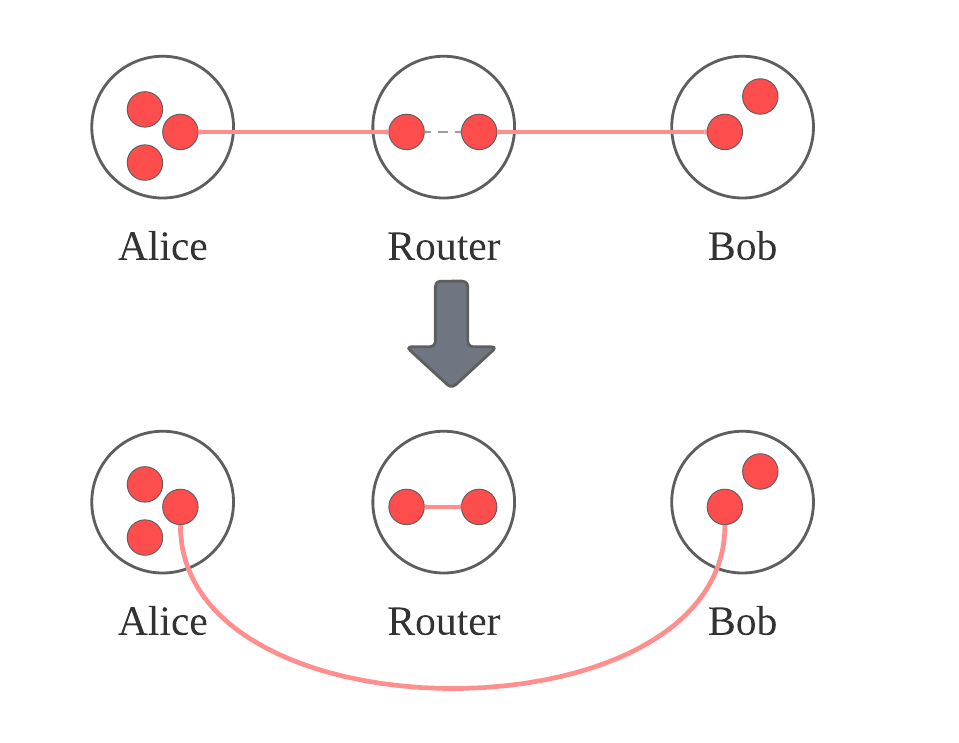}
\caption{Routing Swaps}
\label{fig:background-circuit-examplec}
\end{wrapfigure} 

\noindent\textbf{\textit{Real-world HCQN Protocols.}}
When constructing the HCQNs in real-world networks, more issues are coming out, and researchers investigated different aspects of the real-world HCQNs and came up with different real-world HCQN protocols to overcome the difficulties in certain aspects, such as proposing QPass and QCast \cite{10.1145/3387514.3405853}, to permit long-distance remote quantum communications.
There are three important assumptions about real-world routing networks.
First, locations in an HCQN can be far away, so a direct channel cannot be established between two locations that are far away.
There is usually a routing graph structure to guide the possible quantum channel establishment between two locations.
Far away locations must use QES (\Cref{sec:enswaps}) to prolong a quantum channel so that the two locations are possibly connected.
Second, a quantum message communication between two locations can have many different message transmission paths,
but one path might be the best choice among these paths, based on a choice property, such as the shortest path or maximizing the success probability rate to transmit the message.
The success rate of message transmission can be below $1$ because quantum computers are noisy, so the quantum channel has a certain chance of failure for conveying a quantum message.
Finally, quantum channels have a short lifetime, so it is unlikely that an established quantum channel can be stalled for a long time to wait to transmit a message.
Thus, many real-world protocols assume a set of quantum messages are waiting to be sent, and a quantum channel is established to send a pre-defined quantum message; our Intention IDs and Intention ID maps capture such behavior.

 \section{Superdense Coding Example}\label{sec:superdense}

\begin{figure}[h]
{\centering
    \includegraphics[width=0.6\textwidth]{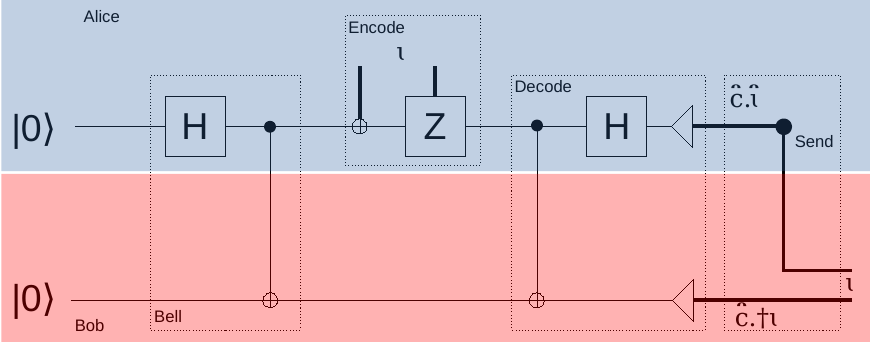}
            \caption{Superdense Coding Circuit\newline}
            \label{fig:background-circuit-supera}
}
\end{figure}

Superdense coding is the protocol to encode a classical message into a quantum channel, which can then be used as a quantum message.
It is the example usage of rules \rulelab{Com} and \rulelab{CLocal} is the circuit implementation of the superdense coding protocol, shown in \Cref{fig:background-circuit-supera}  \footnote{Superdense coding assumes to have an extra step of physically sending a qubit from Alice to Bob in the middle of the circuit. In the circuit implementation, we remove the step and let Alice send a classical message to Bob in the end.}, which shows how a piece of classical information can be encoded into a quantum channel and decodes the channel to extract the classical information content. 

\begin{example}[Superdense Coding Transitions]\label{def:example5}

{\small
\[
\begin{array}{ll}
&\pard{\bscell{\pard{\seq{\downa{c}}{\seq{ \encode{c}{\iota}}{\seq{\decode{c}{x}}{\csenda{a}{x}0}}}}{\circ}}}
{\bscell{\pard{\seq{\seq{\downa{c}}{\seq{\creva{c}{y}}{\seq{\creva{a}{z}}{\textcolor{red}{\encode{y}{z}}}}}}{\textcolor{red}{\creva{y}{w}0}}}{\circ}}}
\\[0.5em]
\xrightarrow{c}
&\pard{\bscell{\pard{\seq{ \encode{c}{\iota}}{\seq{\decode{c}{x}}{\csenda{a}{x}0}}}{c.\circ}}}
{\bscell{\pard{\seq{{\seq{\creva{c}{y}}{\seq{\creva{a}{z}}{\encode{y}{z}}}}}{\creva{y}{w}0}}{c.\circ}}}
\\[0.5em]
\longrightarrow
&\pard{\bscell{\pard{\seq{\decode{c}{x}}{\csenda{a}{x}0}}{c.\iota}}}
{\bscell{\pard{\seq{{\seq{\creva{c}{y}}{\seq{\creva{a}{z}}{\encode{y}{z}}}}}{\creva{y}{w}0}}{c.\circ}}}
\\[0.5em]
\equiv
&\parll{\bscell{{{\seq{\decode{c}{x}}{\csenda{a}{x}0}}}}}{c.\iota}
{\bscell{\seq{{\seq{\creva{c}{y}}{\seq{\creva{a}{z}}{\encode{y}{z}}}}}{\creva{y}{w}0}}}
\\[0.5em]
\xrightarrow{c.\iota}
&\pard{\bscell{{{{\csenda{a}{\cmsg{c}.\cmsg{\iota}}0}}}}}
{\bscell{\pard{\seq{{{\seq{\creva{a}{z}}{\encode{\cmsg{c}}{z}}}}}{\creva{\cmsg{c}}{w}0}}{\cmsg{c}.\up{\iota}}}}
\\[0.5em]
\xrightarrow{a.\cmsg{c}.\cmsg{\iota}}
&
\textcolor{red}{
\pard{\bscell{{{{0}}}}}
{\bscell{\pard{\seq{{{{\encode{\cmsg{c}}{\cmsg{c}.\cmsg{\iota}}}}}}{\creva{\cmsg{c}}{w}0}}{\cmsg{c}.\up{\iota}}}}
}
\\[0.5em]
\longrightarrow
&
\textcolor{red}{
\pard{\bscell{{{{0}}}}}
{\bscell{\pard{\creva{\cmsg{c}}{w}0}{\cmsg{c}.(\cmsg{c}.\cmsg{\iota}\odot\up{\iota})}}}
}
\\[0.5em]
\equiv
&\pard{\bscell{{{{0}}}}}
{\bscell{\pard{\creva{\cmsg{c}}{w}0}{\cmsg{c}.\iota}}}
\\[0.5em]
\longrightarrow
&\pard{\bscell{{{{0}}}}}
{\bscell{0}}
\end{array}
\]
}
\end{example}

We show the QAM configuration and transitions above.
In the final step, Alice and Bob both receive classical residues $\cmsg{c}.\cmsg{\iota}$ and $\cmsg{c}.\up{\iota}$,
meaning that the classical information is split into two messages.
After that, two extra steps are marked red in \Cref{def:example5} but not shown in \Cref{fig:background-circuit-supera}. The first extra step merges the two classical residues with the original classical information by applying an encoding operation. After that, the encoded classical information exists as a resource molecule in the system. Since its content is classical, we can apply rule \rulelab{CLocal} to localize the classical information, being used as data in the local processes.

In the marked black part of the superdense coding transitions,
the first few rule applications are similar to the ones in the quantum teleportation transitions above.
After the decoding rule, we have the classical residue $\cmsg{c}.\cmsg{\iota}$ for Alice (the left membrane), and another classical residue $\cmsg{c}.\up{\iota}$, living as a molecule, for Bob (the right membrane).
Alice then sends $\cmsg{c}.\cmsg{\iota}$ to Bob, and Bob encodes it to the molecule $\cmsg{c}.\up{\iota}$.
The encoding step is a classical one. It combines the two classical residues and recovers the message $\iota$.
Finally, Bob applies rule \rulelab{CLocal} to localize the result classical message $\iota$.

\section{Theorem Proof Outlines} \label{sec:qamsemproof1}

Here, we provide an outline of theorems in the paper.

\subsection{No-cloning Theorem Proof} \label{sec:qamsemproof1a}

The first theorem is \Cref{def:no-cloning-good}, restated below.

\begin{theorem}[No-cloning Assurance]\label{def:no-cloning-gooda}\rm

For any membrane $\scell{\pard{\overline{R}}{\overline{\phi}}}$ in a QAM configuration $\overline{P}$, for any quantum message $q$ located as a resource molecule in $\overline{\phi}$, any transition $\overline{P} \to \overline{Q}$ does not substitute and copy $q$ in $\overline{P}$ more than once.
\end{theorem}

\begin{proof}\label{def:no-cloning-goodg}\rm
The proof is done by rule induction on the QAM semantic rules. There are five cases involving substitution that might copy a quantum message $q$ more than once and we analyze them one by one as follows:

\begin{itemize}

\item The application of rule \rulelab{QLocal} substitutes a specific channel $\alpha$ with a quantum message $\alpha.q$ labeled with the same channel $\alpha$. Clearly, it does not create a copy of the message $q$.

\item The application of rule \rulelab{CLocal} substitutes variable $x$ with a classical message $\iota$, so it does not involve quantum messages.

\item The application of rule \rulelab{Encode} computes a meet operation between the existing channel content $\mu_1$ in the channel $\alpha$ and a new message $\mu_2$ that might be a quantum message.
However, the meet operation does not involve substitution, so it does not create a copy of the message $\mu_2$.

\item The application of rule \rulelab{Com} performs the communication of a classical message $\iota$, so it does not involve quantum messages. 

\item The application of rule \rulelab{MT} duplicates a process $R$, but $R$ is a process, and all quantum messages are assumed to be as a resource module, so the duplication does not create a copy of a quantum message as a resource molecule.

\end{itemize}
\end{proof}

\subsection{Non-Relocation Theorem Proof} \label{sec:qamsemproof1b}

Next, we show the outline for proving \Cref{def:decoding-good}.
The proof of the theorem relies on several lemmas as well as a well-formedness assumption on all channels.
We first state the well-formedness assumption.
In QAM, the label of every quantum or projective channel to be created is uniquely named, as described by the following well-formedness definition.

\begin{definition}[Channel Well-formedness]\label{def:well-formed}\rm
A QAM configuration $\overline{P}$ is well-formed if the following two conditions are satisfied:

\begin{itemize}
\item A quantum/projective channel in membrane resource molecules appears no more than twice.
\item Every pair of channel creation operations creates a unique quantum channel name than other channels appearing in other pairs and membrane resources.
\end{itemize}
\end{definition}

Given a QAM configuration $\overline{P}$, if it satisfies the well-formedness definition, the evaluation of the configuration does not result in an ill-formed state.
We now show a lemma below, stating that only a decoder in a QAM configuration can erase a quantum resource from a membrane:

\begin{lemma}[Decoding Uniqueness Property]\label{def:decoding-gooda1}\rm
For or any membrane $P=\scell{\pard{\overline{R}}{\overline{\phi}}}$ in a QAM configuration $\overline{P}$, for a quantum resource labeled with $c$ in $\overline{\phi}$,
only a decoding operation can erase $c$ from the membrane $P$.

\end{lemma}

\begin{proof}\label{def:decoding-gooda1-pf}\rm
The proof is done by rule induction on the QAM semantic rules. There are five cases involving a quantum channel $c$ and we analyze them one by one as follows:

\begin{itemize}

\item The application of rule \rulelab{QLocal} substitutes a specific channel $\alpha$ with a quantum message $\alpha.q$ labeled with the same channel $\alpha$ in the process $R$. After the rule application, the quantum resource labeled with the channel $c$ is still in the membrane $P$, since the process $R$ applying rule \rulelab{QLocal} is in $P$.

\item The application of rule \rulelab{CLocal} substitutes variable $x$ with a classical message $\iota$, so it does not involve quantum channels.

\item The application of rule \rulelab{Encode} might push a quantum resource labeled with channel $c$ to another quantum resource $\alpha$, but $\alpha$ is located in $P$, so the channel $c$ is still in $P$.

\item The application of rule \rulelab{Com} performs the communication of a classical message $\iota$, so it does not involve quantum messages. 

\item The application of rule \rulelab{Decode} can erase a channel $c$, but it is valid in our lemma. 

\end{itemize}
\end{proof}

We then need to show that an encoding operation can turn a pair of correlated classical and quantum residues into a quantum resource labeled with a proper channel.

\begin{lemma}[Quantum Channel Generation]\label{def:decoding-gooda2}\rm
For or any membrane $P=\scell{\pard{\overline{R}}{\overline{\phi}}}$ in a QAM configuration $\overline{P}$, to generate a quantum channel $c$ in the membrane $P$,
we can either perform a quantum channel creation operation; or an encoder where it is performed by a process $R=\seq{\encode{\alpha}{\alpha.\cmsg{c.\mu}}}{R'}$ and the resource being encoded is $\phi=\alpha.\up{c.\mu}$ with $\alpha$ being projective channels.

\end{lemma}

\begin{proof}\label{def:decoding-gooda2-pf}\rm
The proof is done by rule induction on the QAM semantic rules with case analyses.
\end{proof}

We are now ready to prove the whole theorem below.

\begin{theorem}[Non-Relocation]\label{def:decoding-gooda}\rm
For any $P=\scell{\pard{\overline{R}}{\overline{\phi}}}$ in a QAM configuration $\overline{P}$, for a quantum resource labeled with $c$ in $\overline{\phi}$,
if it is transformed to another membrane $Q$ in $\overline{P}$, then

\begin{itemize}

\item There must be a decoding operation to erase $c$ in $P$ first.
\item After the decoding, there is an encoding in $Q$ to reconstruct $c$ in $Q$'s resource fields.

\end{itemize}
\end{theorem}

\begin{proof}\label{def:decoding-gooda-pf}\rm
By \Cref{def:decoding-gooda1}, we can see that only a decoder can allow erasing the channel $c$. The other rule applications do not erase the channel so it is enough to only discuss the case when a decoding operation is applied. In this case, we can analyze the result of applying a decoder and see that for any quantum channel $c$ or $\alpha.c$, after applying a decoder, it creates projective channels $\cmsg{c}$ or $\cmsg{alpha}$, as well as two residues, are produced.

With the above assumption, we can prove the theorem by rule step induction on the number of steps for applying a QAM semantic rule on a QAM configuration $\overline{P}$.
To create a channel $c$ in Q, by \Cref{def:decoding-gooda2}, there are only two possible operations: a channel creation or an encoder.

\begin{itemize}

\item The operation cannot be a channel creation operation because of the well-formedness assumption. We have already had the channel $c$ existed in the system, any newly created channel cannot have the same name.

\item If the operation is an encoder, by \Cref{def:decoding-gooda2}, it then reconstructs a quantum resource labeled with the channel $c$. 

\end{itemize}
\end{proof}

\subsection{Refinement Theorem Proof} \label{sec:qamsemproof1c}

We now show the theorem that connects QPass and QCast below.

\begin{theorem}[QPass and QCast Trace Refinement]\label{def:traceeq2a}\rm
Given $\xi$, $\theta(\overline{p})$, and an initial configuration $C\in \Cs_d$ and the QPass and QCast systems as stated in \Cref{sec:case-qpass},
there is an assignment $\overline{p}$, such that (QPass,$C$) $\sqsubseteq_r$ (QCast,$C$).

\end{theorem}

\begin{proof}\label{def:traceeq2a-pf}\rm
The proof of \Cref{def:traceeq2} is done by rule induction (QPass) on the configuration of $C$, for any QPass transition on $C$, we can uniformly set the probabilistic success rate $\overline{p}$ to be all $0.5$, in this case, the maximized success rate path and the shortest path for any path between two locations (defined in $\xi$) are the same. Thus, the predicates $\varphi_p$ and $\varphi_c$ are the same, so that (QPass,$C$) $\sqsubseteq_r$ (QCast,$C$).

\end{proof}

\section{Extend the QAM for Security Protocols} \label{sec:qamsecurity}

\newcommand{\tget}{\texttt{get}}
\newcommand{\tstart}{\texttt{start}}
\newcommand{\tfst}{\texttt{fst}}
\newcommand{\tsnd}{\texttt{snd}}
\newcommand{\tucom}[1]{\texttt{ucom}~{#1}}
\newcommand{\tif}{\texttt{if}}
\newcommand{\tthen}{\texttt{then}}
\newcommand{\telse}{\texttt{else}}
\newcommand{\tlet}{\texttt{let}}
\newcommand{\tin}{\texttt{in}}
\newcommand{\transs}[3]{[\!|{#1}|\!]^{#2}_{#3}}

Many quantum security protocols, such as BB84 \cite{BENNETT20147}, have built in the concept of encrypting the quantum state by using some additional quantum gates, such as Pauli gates below:

{
\begin{center}
$
    X\,i = \begin{pmatrix}0 & 1\\1 & 0\end{pmatrix},\quad
    Y\,i = \begin{pmatrix}0 & -i\\i & 0\end{pmatrix},\quad
    Z\,i = \begin{pmatrix}1 & 0\\0 & -1\end{pmatrix},
$
\end{center}
}

Here, The index $i$ in the gates refers to the index position of applying the gate on a qubit.
In the QAM, we can extend the encoding operation $\encode{\alpha}{\mu_1}$ to support the security encryption concept by permitting an operation as $\encode{\alpha}{G\,i}$ where $G\in\{X,Y,Z\}$.
This operation means that we apply gate $G\,i$ on the $i$'s position of the quantum resource state labeled as $\alpha$.
The semantics of the encoding operation will remain the same, but we need to modify the operation in the meet operation $\odot$, so that we can apply a gate $G\,i$ as follows:

{
\begin{mathpar}

   \inferrule[Encode]{}
       {\scell{\pard{\pard{\seq{\encode{\alpha}{G\,i}}{R}}{\alpha.\mu_2}}{...}}
             \longrightarrow {\scell{\pard{\pard{R}{\alpha.(G\,i(\mu_2))}}{...}} }}
  \end{mathpar}
}

The modified $\odot$ operation is resposible for applying $G\,i$ on the quantum state $\mu_2$ as $\alpha.(G\,i(\mu_2))$.
With the modification, we will be able to define the security encryption such as the encryption in BB84.
We will also need to extend the decoding to support the measurement in different bases.
The support of analysis for different quantum security protocols will be a future study of the QAM, and we show the possibility here by showing how quantum gates can be merged in the meet operation in the QAM.

\section{The QAM and Quantum Circuits} \label{sec:qamsemproof}

We intend to define the QAM to capture the HCQN behaviors; thus, we can expect that the QAM is instantiated to a quantum programming language \footnote{We do not intend to design the QAM as a programming language, but a process algebra for defining HCQN behaviors with programming language features.}.
Here, we show an instance of the QAM can be compiled to a concurrent quantum circuit language, Concurrent \sqir \cite{VOQC,pythonsqir,sdnquantum}, which is largely built on top of qCCS \cite{10.1145/1507244.1507249} and CQP \cite{10.1145/1040305.1040318}.
The compilation essentially translates an instance of the QAM to a multi-threaded language with the \sqir as a quantum circuit library.
Essentially, Concurrent \sqir can be viewed as a multiset of concurrent processes arranged as:

{
\begin{center}
$
S_1\mid S_2 \mid S_3 \mid ... \mid S_n.
$
\end{center}
}

Each process $S_i$ is a sequential circuit program, each operation of which can be either a quantum circuit operation written in \sqir \cite{VOQC} or a classical operation, such as a classical C-like synchronizer.
Each process also contains a location name indicating the membrane it belongs to.

We compile a specific instance of the QAM to Concurrent \sqir by instantiating abstract syntactic entities in the QAM.
We first instantiate the QAM membranes and quantum resources mentioned in a QAM program with fixed names (identifiers)
that distinguish these entities with each other, such as $\scell{\overline{M}}_g$ (\Cref{fig:q-pi-semantics2}).
Additionally, Quantum channels ($c$) are instantiated as objects with two location names $g$ and $h$ (we use $\cn{lft}(c)$ for the first location and $\cn{rht}(c)$ for the second location.), identifying the two membranes that a quantum channel is connecting.
Essentially, locations refer to quantum qubit array regions. For example, in quantum teleportation in \Cref{def:example1}, the total number of qubits is $3$ so the qubit array in the system is a length-$3$ qubit array. Alice's two qubits are in the region $[0,1]$, while Bob's qubit is in the region of $[2,3)$. 
Then, The membrane location for Alice refers to the index $0$, the starting position of Alice's qubit region, and the membrane location for Bob refers to the index $2$.

Since the QAM molecules can be either processes or quantum resources, the QAM compiler generates not only a concurrent program but also quantum states. We simplify the compilation by assuming only one compiled target quantum state $\varphi$, and every quantum operation in the compiled concurrent program is applied on some qubit locations in $\varphi$. 
In the compilation, we translate the QAM operations that manipulate quantum resources, such as quantum channel creations, encoding quantum messages, and quantum channel decoding, to quantum circuits written in \sqir, classical operations that manipulate classical resources to classical functions, and quantum resources to a quantum state. 
Each named quantum resource is translated to a qubit array that resides as a section in the state,
where we ensure any two different quantum resources are disjoint.
The QAM compiler is expressed as the following two judgments:

{
\begin{center}
$
\Omega;\Sigma;g\vdash R \gg S
\qquad
\Omega;\Sigma\vdash \overline{P} \gg (\overline{S},\varphi)
$
\end{center}
}

The first judgment translates a QAM process to a Concurrent \sqir process, and the second translates a QAM program to a multiset of Concurrent \sqir processes and a quantum state.
Here, $\Omega$ is a type environment mapping from the QAM element parameters, including message names, different channel names, and program variables, to $C$ or $Q$ representing if a parameter is classical or quantum;
$\Sigma$ maps parameters to a pair of two numbers $i$ and $n$ (the first element, $\cn{fst}$, in the pair is $i$ and the second element, $\cn{snd}$, in the pair is $n$). If $g$ and $h$ are the two locations of a channel $c$, then $\Sigma(c.g)$ and $\Sigma(c.h)$ produce the $i$ and $n$ values for the two-qubit location indices of $c$ in $\varphi$, respectively.
$\overline{P}$ is a QAM configuration, $\overline{S}$ is a concurrent multiset of QSDN processes, either a \sqir quantum circuit or a classical function $f$.

We assume that a given global number, $N$, acts as the bandwidth number for a piece of quantum information.
In \Cref{def:example1}, we teleport a unit piece of quantum message $\cmsg{d}.e$, which is not necessarily a single qubit quantum message. In the compilation, we allow users to specify $N$ representing the number of qubits in a unit piece of quantum message, such as $\cmsg{d}.e$. A composed quantum message, such as $c.e$, is proportional to the fixed bandwidth number.
The type environment $\Omega$ determines if a parameter is classical or quantum. 
The compilation procedure fails if a parameter in $R$ and $\overline{P}$ violates the type given in $\Omega$.
$\Sigma$ is an environment recording the qubit bandwidth for a parameter, mapping from parameters to two numbers $i$ and $n$. The first number defines the starting qubit array position for the parameter, and the second number defines the length of the qubit array, proportional to $N$.
Here, we write $\Sigma(x)$ to mean that the starting position $i$ for the qubit array of $x$, and $|\Sigma(x)|$ to refer to the length $n$ of the array.

\ignore{
\begin{figure}[t]
{\centering
\hspace*{0.2em}

\begin{minipage}[b]{.3\textwidth}
                 \includegraphics[width=0.9\textwidth]{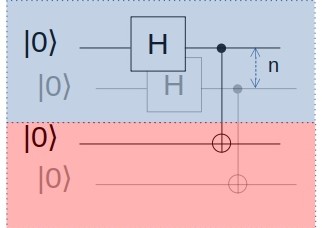}
            \caption{Cohere circuit}
            \label{fig:cohere-circuit}
 \end{minipage}
\hfill
\begin{minipage}[b]{.3\textwidth}
\hspace*{0.2em}
                 \includegraphics[width=0.9\textwidth]{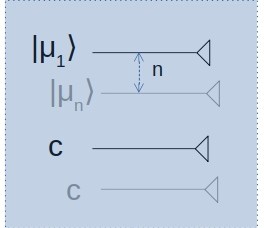}
            \caption{Decode quantum message circuit}
            \label{fig:decodeq-circuit}
 \end{minipage}
 \hfill
\begin{minipage}[b]{.3\textwidth}
\hspace*{0.2em}
                 \includegraphics[width=0.9\textwidth]{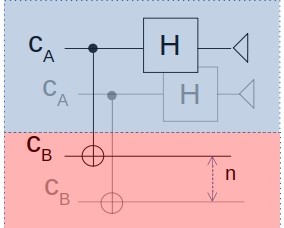}
            \caption{Decode classic message circuit}
            \label{fig:decodec-circuit}
 \end{minipage}
\begin{minipage}[b]{.3\textwidth}
\hspace*{0.2em}
                 \includegraphics[width=0.9\textwidth]{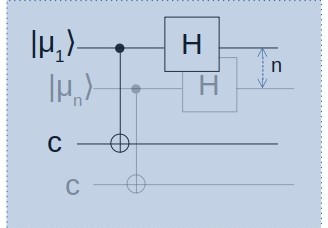}
            \caption{Encode quantum message circuit}
            \label{fig:encode-q-circuit}
 \end{minipage}
 \hfill
\begin{minipage}[b]{.3\textwidth}
\hspace*{0.2em}
                 \includegraphics[width=0.9\textwidth]{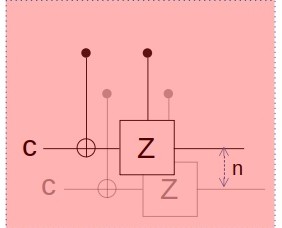}
            \caption{Encode classic message circuit}
            \label{fig:encode-c-circuit}
 \end{minipage}
}
\caption{QAM to circuit compilation}
\label{fig:compilation-circuit}
\end{figure}
}

\begin{figure}[t]
{\footnotesize
  \begin{mathpar}
    \inferrule[C-CohereL]{\cn{lft}(c)=g\\\Sigma(c.g) = i \\ \Sigma(c.\cn{rht}(c)) = j\\\\|\Sigma(c.g)|=|\Sigma(c.\cn{rht}(c))|=n\\ \Omega;\Sigma;g \vdash R \gg S}
        {\Omega;\Sigma;g \vdash \seq{\downa{c}}{R}
             \gg (\forall x \in [0, n) \Rightarrow H\,(i+x);CX\,(i+x)\,(j+x));S}    

    \inferrule[C-CohereR]{\cn{rht}(c)=g\\\Omega;\Sigma;g \vdash R \gg S}
        {\Omega;\Sigma;g \vdash \seq{\downa{c}}{R} \gg S}   

     \inferrule[C-EncodeC]{\Omega(\mu) = C \\ \Sigma(c.g) = i \\|\Sigma(c.g)|=n \\ \Omega;\Sigma;g \vdash  R \gg S}
        {\Omega;\Sigma;g\vdash \seq{\encode{c}{\mu}}{R} \gg (\forall x \in [0, n) \Rightarrow CX\ \mu[x]\ (i+x);CZ\ \mu[x]\ (i+x));S} 

    \inferrule[C-EncodeQ]{\Omega(\mu) = Q \\ \Sigma(c.g) = i \\|\Sigma(c.g)|=n \\ \Omega;\Sigma;g \vdash  R \gg S}
        {\Omega;\Sigma;g \vdash {\seq{\encode{c}{\mu}}{R}} \gg 
      (\forall x \in [0, n) \Rightarrow H\ (i+x);CNOT\ \mu[x]\ (i+x));S}   

    \inferrule[C-DecodeQ]{\Omega(\mu) = Q \\ \Sigma(c.\cn{lft}(c)) = i \\  \Sigma(c.\cn{rht}(c)) = j \\ \Omega;\Sigma;g \vdash R \gg S}
        {\Omega;\Sigma;g \vdash {\seq{\decode{c}{x}}{R}} \gg
             (\forall x \in [0, n) \Rightarrow Meas\ (i+x));\cn{send}(c,\cmsg{c});S} 

    \inferrule[C-Rev]{ \Omega;\Sigma;g \vdash R \gg S}
        {\Omega;\Sigma;g \vdash {\seq{\creva{\delta}{x}}{R}} \gg \cn{wait}(\delta,x);S} 

    \inferrule[C-MemP]{ \Omega;\Sigma;g \vdash R \gg S'\\\Omega;\Sigma \vdash {\scell{{\overline{M}}}_g}\gg(\overline{S},\varphi)}
        {\Omega;\Sigma \vdash {\scell{\pard{R}{\overline{M}}}_g} \gg (S'\mid\overline{S},\varphi)} 
  \end{mathpar}
}
\caption{QAM compilation rules. $\Sigma(x)=\cn{fst}(\Sigma(x))$ and $|\Sigma(x)|=\cn{snd}(\Sigma(x))$.}
\label{fig:compilation-rules}
\end{figure}

To see why $\Sigma$ is needed, recall that in the teleportation example ($\Cref{def:example1}$),
we encode a message $\cmsg{d}.e$ to a quantum channel $c$. Assume that the message $\cmsg{d}.e$ (more specifically, parameter $e$) is translated to a quantum array having length $N$. This means that the channel $c$ must be at least translated to a qubit array of length $2N$ to be able to hold the message; otherwise, the quantum teleportation execution would fail.
In addition, assume that we instead would like to encode a portion of a channel $d.e$ into the quantum channel $c$. If $e$ is a length $N$ array, a portion of the quantum channel $d.e$ has a length $2N$, so the channel $c$ is estimated to have length $4N$, which results in a different situation.

In the compilation, $\Sigma$ is given as an input of the compilation. 
Since the QAM describes nondeterminism in HCQN protocols, when executing the compiled code, 
it is possible that a channel $c$ holds a quantum message $q$ that surpasses the defined bandwidth for $c$.
We classify such a case to be a failure state.
The QAM compilation correctness is defined by the assumption that no failure states can be reached.

The compilation rules are given in \Cref{fig:compilation-rules}.
Rules \rulelab{C-CohereL} and \rulelab{C-CohereR} describe the compilation procedure for quantum channel creation,
respectively representing the behaviors of the two processes living in the two membranes holding the two ends of the channel.
In the compilation procedure, we assume that all the channel creation task is handled by one end (the \cn{lft} end) of the channel,
and the \cn{rht} end does nothing. The channel creation procedure does not require a synchronizer because there are no operations in the QAM that need to surely happen before a quantum channel creation.

Rules \rulelab{C-EncodeC} and \rulelab{C-EncodeQ} describe the encoding procedure when the message being encoded is classical and quantum, respectively. The two kinds of encoding have different implementations at the quantum circuit level, even though they can both be viewed as encoding. Here, we also utilize $\Sigma$ to learn about the number of qubits in the target quantum channel to produce the correct circuit with respect to the correct qubit number.
Rule \rulelab{C-DecodeQ} describes the compilation procedure of decoding a quantum channel with a quantum message content.
There is another similar rule describing the decoding of a quantum channel with classical message content.
In the rule, we first look at the qubit positions of the \cn{lft} end of the quantum channel $c$, then place a list of measurement operations on these positions. The operation $\cn{send}(c,\cmsg{c})$ is a classical synchronizer. It sends a message to the other end of the channel $c$ and notifies the other side that the decoding is finished and the other side can move forward.

Rule \rulelab{C-Rev} is a classical receiver operation acting as the receiver synchronizer.
It waits to receive a classical message, either a projective channel name or a classical message, from the channel $\delta$.
For example, once a \rulelab{C-DecodeQ} rule is applied on a membrane $g$, it sends out a projective channel $\cmsg{c}$ via the quantum channel name $c$. The other end of the channel waits to receive the projective channel. The compiled code synchronizes the \cn{send} and \cn{wait} operations to ensure they happen simultaneously.
Finally, rule \rulelab{C-MemP} describes the compilation of a QAM program based on the compilation procedures of sub-processes.

The QAM compilation correctness is based on a version of the simulation relation with the exclusion of the failure state listed below.
We implement the QAM compilation in Coq and output the code to Ocaml. Then, we test the correctness proposition by running programs in a QAM simulator against the execution of the compiled \sqir program.

\begin{theorem}\label{thm:vqir-compile}\rm[QAM translation correctness]
  Suppose $\overline{P}\xrightarrow{\kappa} \overline{Q}$,
  $\Omega;\Sigma\vdash \overline{P} \gg (\overline{S},\varphi)$ and $\Omega;\Sigma\vdash \overline{Q} \gg (\overline{S'},\varphi')$.
then there is a non-failure \sqir execution, such that $(\varphi,\overline{S}) \longrightarrow^* (\varphi',\overline{S'})$.
\end{theorem}

\subsection{A Proof Outline of Theorem \Cref{thm:vqir-compile}}\label{appx:main}

The theorem is proved by rule induction on the semantic rules of QAM. We assume that the translated circuit allows an $n$ qubit width channel to convey quantum messages. 

\myparagraph{Channel Creation} For the \rulelab{Cohere} rule, assuming that processes $\overline{M}$, $\overline{N}$, $R$ and $T$, as well as the other membranes not mentioned in the system are properly translated, the rule application is valid only if we have the following conditions:

\begin{enumerate}

\item The pair of channel names $c$ (the two ends of channel $c$) is uniquely defined in the system, and the two ends of channels are classified as $\texttt{lft}$ and $\texttt{rht}$.

\item There are two pieces of available quantum resources $\circ$, which are empty at the beginning.

\end{enumerate}

The translated circuit has a generalized Bell pair circuit (possibly $n$ qubit width) connecting the pair $\texttt{lft}(c)$ and $\texttt{rht}(c)$. With the availability, the system properly creates a generalized Bell pair as a quantum channel.

\myparagraph{Quantum State Equations}
Every equation appearing in the \rulelab{AirEQ} defines an equivalence property appearing in a quantum state; thus, they are automatically valid in a quantum state.

\myparagraph{Quantum Encoding Operations} For any QAM program $P$, rules \rulelab{QLocal} and \rulelab{CLocal} are substitution rules. Their correctness depends on the inductive step of proving the substituted variables in $P$ with the given subterm. Rule \rulelab{Encode} installs a quantum message into a channel, which has two different cases:

\begin{enumerate}

\item If the message encoded is a classical message, which is verified by the predicate $\Omega(\mu)=C$ in rule \rulelab{C-EncodeC}, the generated circuit properly installs the classical message $\mu$ into a quantum channel.

\item If the message encoded is a quantum message, which is verified by the predicate $\Omega(\mu)=Q$ in rule \rulelab{C-EncodeQ}, the generated circuit properly installs the message $\mu$ into a quantum channel $c$, provided that the width of the quantum channel $c$ is large enough to hold the message.

\end{enumerate}

The final result of the encoding is defined by the operation $\odot$, whose equational rules are validated by quantum state properties in Concurrent SQIR. 

\myparagraph{Quantum Decode Operations}
Assuming that processes $\overline{M}$, $\overline{N}$, $R$ and $T$, as well as the other membranes not mentioned in the system are properly translated, the application of the \rulelab{Decode} rule can be divided into two cases:

\begin{enumerate}

\item If the message decoded is a classical message, which is verified by the predicate $\Omega(\mu)=C$ in rule \rulelab{C-DecodeC}, the generated circuit applies a series of $CX$ ans $CZ$ gates, with a measurement operation in the end. It extracts a classical message out of a quantum channel.

\item If the message encoded is a quantum message, which is verified by the predicate $\Omega(\mu)=Q$ in rule \rulelab{C-DecodeQ}, we then apply two series of measurement operations on the two ends of a channel $c$. The result is the same as the measurement result of a quantum teleportation circuit, which creates a pair of a classical residue $\cmsg{c}.\cmsg{q}$ and a quantum residue $\cmsg{c}.\up{q}$, such that the combination of the two residues recovers the message $q$.

\end{enumerate}

\myparagraph{Other Operations}
The classical message communication is translated to a classical message communication, similar to the $\Pi$-calculus communication operation that we define on top of Concurrent SQIR, the choice operation is also a classical $\Pi$-calculus style choice operation that we define on top of Concurrent SQIR, and the Concurrent SQIR semantics of the replication operations are also from $\Pi$-calculus. The other operations are syntactic rules appearing in the QAM level and they are properly translated to equivalent processes in Concurrent SQIR.

\end{document}